\documentclass[a4paper,unpublished,onecolumn,11pt]{quantumarticle}
\pdfoutput=1

\usepackage[utf8]{inputenc}
\usepackage[T1]{fontenc}
\usepackage[english]{babel}

\usepackage{amsmath,amssymb,amsthm,bm}
\usepackage{physics}

\usepackage{graphicx}
\usepackage{booktabs}
\usepackage{siunitx}
\usepackage{enumitem}

\usepackage[table]{xcolor}
\usepackage[most]{tcolorbox}
\usepackage{mdframed}

\usepackage{tikz}
\usepackage{quantikz}

\usepackage{pgfplots}
\pgfplotsset{compat=1.18}
\usepgfplotslibrary{groupplots}

\usepackage{algorithm}
\usepackage{algpseudocode}

\usepackage[
  backend=biber,
  style=numeric,
  sorting=none
]{biblatex}
\usepackage{hyperref} 
\newcommand{\eps}{0.1}

\AtBeginEnvironment{abstract}{\raggedright}

\newtheorem{theorem}{Theorem}
\newtheorem{lemma}[theorem]{Lemma}
\newtheorem{proposition}[theorem]{Proposition}
\newtheorem{corollary}[theorem]{Corollary}
\theoremstyle{definition}
\newtheorem{definition}[theorem]{Definition}
\newtheorem{assumption}[theorem]{Assumption}
\theoremstyle{remark}
\newtheorem{remark}[theorem]{Remark}

\newcommand{\poly}[1]{\mathrm{poly}\bigl(#1\bigr)}

\newcommand{\OH}{\mathcal{H}_{\mathrm{OH}}}

\newcommand{\dist}{\operatorname{dist}}
\newcommand{\ZOH}{\mathcal Z_{\mathrm{OH}}}

\makeatletter
\g@addto@macro\normalsize{%
    \abovedisplayskip 3pt plus 1pt minus 1pt%
    \abovedisplayshortskip 3pt plus 1pt minus 1pt%
    \belowdisplayskip 3pt plus 1pt minus 1pt%
    \belowdisplayshortskip 3pt plus 1pt minus 1pt%
}
\makeatother

\def\BibTeX{{\rm B\kern-.05em{\sc i\kern-.025em b}\kern-.08em
    T\kern-.1667em\lower.7ex\hbox{E}\kern-.125emX}}

\author{Chinonso Onah}
\affiliation{Volkswagen AG, Berliner Ring 2, 38440 Wolfsburg, Germany}
\affiliation{Department of Physics, RWTH Aachen University, 52056 Aachen, Germany}
\email{Corr. Author. conah@atom-computing.com}

\author{Kristel Michielsen}
\affiliation{Forschungszentrum Jülich, D-52425 Jülich, Germany}
\affiliation{Universit\"at zu K\"oln, 50923 K\"oln, Germany}

\title{Exposing Finite-Depth and Finite-Shot Success Guarantees in Constrained Quantum Optimization via Fej\'er Filtering}

\begin{document}
\maketitle

\begin{abstract}
Finite-depth constrained quantum optimization algorithms need quantitative guarantees that connect circuit resources to the probability of actually sampling feasible or optimal solutions in finitely many shots. We establish such a connection by exposing a positive sampling law in which mixer-driven exploration and spectral selection can be controlled separately. We show that after removing interference between distinct cost eigenspaces as an analytic device, the measurement distribution becomes the normalized product of a mixer-induced exploration envelope and a Fej\'er spectral weight, with the former describing how the mixer spreads probability over the encoded manifold and the latter enhancing the target cost phase while suppressing spectrally separated nontarget phases. In this model, finite-shot success becomes a tractable competition between target weight and off-target leakage, yielding an explicit lower bound on the probability of sampling an optimum.

For the primary bound, we globally rescale the cost Hamiltonian to an integer-valued spectrum, placing the wrapped cost phases on a controlled lattice for Fej\'er filtering. We then define $\delta$ as the minimum circular separation between the optimal phase and every nontarget phase. The single-shot success probability $q_0$ satisfies
\[
q_0
\ge
\frac{x}{1+x},
\qquad
x
=
(p+1)^2
\sin^2\!\left(\frac{\delta}{2}\right)
C_{\beta},
\]
where $p$ is the filter order and $C_{\beta}$ is the mixer-envelope mass on the optimal set. The bound exposes a finite-resource compensation law in which weaker phase separation or smaller envelope mass can be compensated by increased filter order and additional shots.

The same positive-filtering principle exposes a straightforward feasibility guarantee when applied to penalty phases, where discrete phase separation arises naturally. We further show that phase-scale averaging recovers analogous dimension-independent bounds for nonlattice spectra through off-target suppression, extending our results beyond exact lattice normalization. Finally, we construct a postselected coherent spectral filter that reproduces the same Fej\'er weighting, while non-postselected, hardware-efficient implementations remain open.
\end{abstract}

\section{Introduction}
\label{sec:introduction}

We study feasibility and optimality guarantees for finite alternations of the
\emph{Constraint--Enhanced Quantum Approximate Optimization Algorithm} (CE--QAOA)
introduced in Ref.~\cite{onahce}. By construction, CE--QAOA is a shallow, constraint-aware ansatz that
\emph{natively} operates on a block one-hot manifold via a fixed encoding and a
normalized block--XY mixer. In contrast to standard QAOA~\cite{Farhi2014QAOA} and
several variants~\cite{montanezbarrera2024universalqaoa,BaeLee2024RecursiveQAOA,Finzgar2024QIRO},
the quantum dynamics is confined to a  constrained sector from the
outset. It also builds on constraint-preserving mixer and symmetry inspired approaches
(e.g.~\cite{Hadfield2019AOA,Fuchs2022ConstrainedMixers,tsvelikhovskiy2024symmetries})
by making the problem--algorithm co-design explicit~\cite{Tomesh2021QuantumCodesign}
and exploiting a symmetry-aligned nonstabilizer dynamics\cite{onahce,gopalstabilizer,BaertschiEidenbenz2019}.

\medskip\noindent

\paragraph{Feasibility.} In the present paper, we focus on exposing feasibility and optimality guarantees that can be derived with a finite CE-QAOA depth. We define the feasible set $\mathcal F$ as the encoded
computational-basis configurations with zero penalty energy. See the formal statement in Def.\ref{def:kernel-requirement}. To make quantitative statements, we shall frequently refer to amplitude transport into the feasible or optimal set.  By
\emph{amplitude transport} we mean the redistribution of quantum amplitude among encoded configurations under the mixer dynamics. Our work is motivated by the observation that a shallow ansatz may explore the feasible set poorly when its mixer is misaligned with the connectivity of those configurations or the initial state\cite{He2023A}. CE--QAOA addresses this problem by combining a block one-hot encoding with a normalized block-XY mixer adapted to that encoding. 

\medskip\noindent
\textbf{Dimension-free lower bounds.}
We show that both feasibility and optimality success probability can admit  a \emph{nontrivial lower bound}
\(q_0>0\) that is \emph{dimension-free} and holds at finite depth under a mild phase-separation condition. In the main text, the key analytical device is to ``classicalize'' the circuit by inserting a
cost-basis dephasing/twirling channel between layers (used only for analysis).
This produces an interference-free \emph{reference} process whose measurement law
admits an explicit factorization into a mixer-induced envelope and a \emph{positive, band-limited} trigonometric kernel.  We proceed by constructing a  harmonic schedule  where the kernel is exactly the Fej\'er
kernel, i.e., the squared Dirichlet kernel, acting on the wrapped cost phases\cite{WeisseKPM2006,SteinShakarchi2003Fourier,Katznelson2004}. We call this the cost-dephased reference model.

\medskip
\noindent
Our main result shows that, in this dephased reference model, the computational-basis measurement law is the normalized product of a mixer envelope and a Fej\'er spectral weight (Eq.~\ref{eq:factorize}). The total mixer-envelope contribution to the optimal set is captured by
\(
C_\beta=\sum_{z\in\Omega^\star}W_p(z;\boldsymbol\beta),
\)
while the Fej\'er factor suppresses nontarget contributions according to their wrapped phase separation $\delta$. Combining the Fej\'er off-peak control (Lemma~\ref{lem:fejer-offpeak}) with the
envelope mass \(C_\beta\) yields a dimension-free success bound (Theorem~\ref{thm:dimension-free-success} and Eq.~\ref{eq:q0-ratio}). This formalism suggests an instance-guided planning rule based on the two instance-dependent quantities $\delta$ and $C_\beta$, which enter the success bound through
\[
x=(p+1)^2\sin^2(\delta/2)\,C_\beta.
\]
The resulting guarantee has no explicit dependence on the dimension of the ambient Hilbert space, the encoded space, or the feasible set. We subsequently discuss extension to cost phases beyond lattice normalization in Sec. \ref{subsec:rl-averaging} and Lipschitz/main-lobe arguments that justify  reduced filter order with preserved but modified finite shot guarantees (Sec. \ref{sec:depth-reduction}) .


\subsection{Relation to Prior Work}
\label{sec:background}

\medskip\noindent
\textbf{Constraint–preserving mixers and alternating–operator variants.}
Alternating–operator approaches to quantum optimization enforce constraints by designing mixers that preserve structured subspaces, e.g., one-hot or degree/capacity manifolds \cite{Hadfield2019AOA,Fuchs2022ConstrainedMixers,tsvelikhovskiy2024symmetries}. The present work continues in that tradition by using a normalised block-XY mixer aligned with the one-hot encoding but differs in how we \emph{analyze} shallow circuits. Here, we study the diagonal statistics through a nonnegative, band-limited trigonometric filter acting on \(e^{-i\gamma H_C}\). Prior constraint-preserving variants do not, to our knowledge, derive the measurement law as a product of a mixer envelope with a fixed positive kernel or derive the ensuing ratio-type lower bounds reported in this work.

\medskip\noindent
\textbf{Filtering viewpoints (Fej\'er) and dephasing as an analysis device.}
Spectral filtering through trigonometric or polynomial phase functions is a classical idea in numerical and harmonic analysis, where positive kernels such as those of Jackson and Fej\'er serve as smoothing filters that suppress Gibbs oscillations and provide rigorous off--peak control~\cite{WeisseKPM2006}. Despite its long history in classical approximation theory, this perspective has not appeared in the context of variational quantum algorithms. To our knowledge, no existing analysis of QAOA, or other parameterized quantum circuit protocols has formulated their phase dynamics as a \emph{nonnegative trigonometric filter} acting on the cost spectrum. This work shows that such filters can yield explicit analytic control, finite--depth bounds, and geometric insight into the parameter landscape of variational quantum circuits. For background on dephasing and pinching as standard analysis tools, see ~\cite{Watrous2018,TomamichelFiniteResources}.

\medskip\noindent
\subsection{Constraint–Enhanced QAOA}
\label{sec:ce-qaoa}

We adopt the CE–QAOA kernel introduced in Ref.~\cite{onahce}.  For completeness we restate the definition and the minimal properties (Propositions 2–4) used later on in the analyses.  The construction follows the  \emph{alternating–operator} (``QAOA+'') paradigm\cite{Hadfield2019AOA} which replaces the generic X–mixer by  symmetry–preserving mixers that act invariantly on  constraint projectors, thereby confining evolution to structured subspaces \cite{Hadfield2019AOA,Fuchs2022ConstrainedMixers, tsvelikhovskiy2024symmetries}. 
\emph{Constraint–Enhanced QAOA} (CE–QAOA) \cite{onahce} follows this direction but makes the problem–algorithm co–design explicit by introducing a kernel designed to operate on block one–hot manifolds with a \emph{fixed} mixer and initial state family to match the encoding. The CE–QAOA \emph{kernel} is defined as:

\begin{definition}[CE--QAOA kernel]
\label{def:kernel-requirement}
An optimization instance \(I\) belongs to the \emph{CE--QAOA kernel}
if there exist integers \(n,m\in\mathbb N\) and a one-hot encoder
\(\mathsf E_{\mathrm{1hot}}\) with the following structure. Let
\[
[n]:=\{1,\ldots,n\}.
\]
Each of the \(m\) variable blocks contains \(n\) qubits. For
\(k\in[n]\), define the one-excitation computational-basis state
\[
\ket{e_k}
:=
\ket{0}^{\otimes(k-1)}
\otimes
\ket{1}
\otimes
\ket{0}^{\otimes(n-k)}
\in
(\mathbb C^2)^{\otimes n}.
\]
The one-hot subspace of one block is
\[
\mathcal H_1
:=
\operatorname{span}
\{\ket{e_1},\ldots,\ket{e_n}\}.
\]

The configuration-label set and its encoded basis states are
\[
\ZOH
:=
[n]^m,
\qquad
z=(z_1,\ldots,z_m)\in\ZOH,
\qquad
\ket{z}
:=
\bigotimes_{b=1}^{m}\ket{e_{z_b}}.
\]
The full encoded Hilbert space is therefore
\[
\OH
:=
\operatorname{span}
\{\ket{z}:z\in\ZOH\}
=
(\mathcal H_1)^{\otimes m}.
\]
Thus \(n\) is the number of qubits and encoded symbols per block,
\(m\) is the number of variable blocks, and the full register contains
\(mn\) qubits. Throughout, \(z\in\ZOH\) denotes a
configuration label, whereas \(\ket{z}\in\OH\) denotes the
corresponding computational-basis vector.

The problem Hamiltonian splits as
\[
H_C
=
H_{\mathrm{pen}}
+
H_{\mathrm{obj}},
\]
where \(H_{\mathrm{obj}}\) is the Ising Hamiltonian representing the
objective and is diagonal in the computational basis.
The penalty Hamiltonian \(H_{\mathrm{pen}}\) is also diagonal in the
computational basis and has the additional properties specified below.

\begin{enumerate}[label=\textup{(\alph*)}, leftmargin=2.2em]
\item \emph{Penalty structure.}
      $H_{\mathrm{pen}}$ is a sum of squared affine
      one--hot/degree/capacity penalties (optionally plus linear
      forbids), including the required penalty weights, with integer
      coefficients bounded by $\mathrm{poly}(n)$.

\item \emph{Pattern symmetry.} $H_{\mathrm{pen}}$ is invariant under
      (i) block permutations $S_m$ and (ii) global symbol relabelings $S_n$. Hence the configuration space decomposes into level sets
      $L_t=\{z:\, H_{\mathrm{pen}}(z) :=\langle z \mid H_{\mathrm{pen}} \mid z\rangle=t\}$ that are preserved setwise.       The feasible set is
      the zero-penalty level
      \[
      \mathcal F
      :=
      \{z\in\mathcal H_{\mathrm{OH}}:
      H_{\mathrm{pen}}(z)=0\},
      \]
      and $H_{\mathrm{pen}}(z)>0$ for
      $z\notin\mathcal F$. Consequently,
      \[
      \operatorname{spec}(H_{\mathrm{pen}})
      \subseteq
      \{0,1,\ldots,t_{\max}\},
      \qquad
      t_{\max}=\mathrm{poly}(n).
      \]

      A sufficient penalty condition is
      \[
      \min_{z\notin\mathcal F}H_{\mathrm{pen}}(z)
      >
      \max_{z\in\mathcal H_{\mathrm{OH}}}H_{\mathrm{obj}}(z)
      -
      \min_{z\in\mathcal H_{\mathrm{OH}}}H_{\mathrm{obj}}(z).
      \]
\item In addition, the initial state is the $+1$ eigenstate of the block–local normalized XY mixer Hamiltonian,
      \[
      \widetilde H_{M}^{(b)} \;=\;
      \frac{1}{n-1}\sum_{0\le j<k\le n-1}(X_j^{(b)}X_k^{(b)}+Y_j^{(b)}Y_k^{(b)}),
      \]
      with $\|\widetilde H_{M}^{(b)}\|=O(1)$ on each block. The initial state is the       uniform one–hot product
      \[
      \ket{s_0}\;=\;\ket{s_{\mathrm{b}}}^{\otimes m},
      \qquad
      \ket{s_{\mathrm{b}}}\;=\;\frac{1}{\sqrt n}\sum_{k=1}^{n} \ket{e_k}
      \quad\text{(a $W_n$ state per block)}.
      \]
\end{enumerate}
\end{definition}

Consequently, a depth–$p$ CE–QAOA layer stack is
\[
  \ket{\psi_p(\vec\gamma,\vec\beta)}
  \;=\;
  \Bigl(\prod_{\ell=1}^{p} U_M(\beta_\ell)\,e^{-i\gamma_\ell H_C}\Bigr)\ket{s_0},
  \qquad
  \vec\gamma=(\gamma_1,\dots,\gamma_p),\;
  \vec\beta=(\beta_1,\dots,\beta_p).
\]

\begin{figure}[htbp]
  \centering
  \resizebox{\textwidth}{!}{%
  \begin{quantikz}[row sep=0.1cm, column sep=0.05cm]
  \lstick[wires=4]{$\text{Block }0\; \ket{0}^{\otimes n}$}
      & \gate[wires=4,style={rounded corners,fill=blue!8}]{\texttt{OneHotBlock}}
      & \qw
      & \gate[wires=12,style={rounded corners,fill=green!16}]{U_C(\gamma_1)}
      & \gate[wires=4,style={rounded corners,fill=orange!20}]{U_M^{(0)}(\beta_1)}
      & \qw
      & \gate[wires=12,style={rounded corners,fill=green!16}]{U_C(\gamma_2)}
      & \gate[wires=4,style={rounded corners,fill=orange!20}]{U_M^{(0)}(\beta_2)}
      & \qw
      & \gate[wires=12,style={rounded corners,fill=green!16}]{U_C(\gamma_3)}
      & \gate[wires=4,style={rounded corners,fill=orange!20}]{U_M^{(0)}(\beta_3)}
      & \qw
      & \meter{} \\
  & & \qw
    & \qw & \qw & \qw
    & \qw & \qw & \qw
    & \qw & \qw & \qw
    & \meter{} \\
  & & \qw
    & \qw & \qw & \qw
    & \qw & \qw & \qw
    & \qw & \qw & \qw
    & \meter{} \\
  & & \qw
    & \qw & \qw & \qw
    & \qw & \qw & \qw
    & \qw & \qw & \qw
    & \meter{} \\
  \lstick[wires=4]{$\text{Block }1\; \ket{0}^{\otimes n}$}
      & \gate[wires=4,style={rounded corners,fill=blue!8}]{\texttt{OneHotBlock}}
      & \qw
      & \qw
      & \gate[wires=4,style={rounded corners,fill=orange!20}]{U_M^{(1)}(\beta_1)}
      & \qw
      & \qw
      & \gate[wires=4,style={rounded corners,fill=orange!20}]{U_M^{(1)}(\beta_2)}
      & \qw
      & \qw
      & \gate[wires=4,style={rounded corners,fill=orange!20}]{U_M^{(1)}(\beta_3)}
      & \qw
      & \meter{} \\
  & & \qw
    & \qw & \qw & \qw
    & \qw & \qw & \qw
    & \qw & \qw & \qw
    & \meter{} \\
  & & \qw
    & \qw & \qw & \qw
    & \qw & \qw & \qw
    & \qw & \qw & \qw
    & \meter{} \\
  & & \qw
    & \qw & \qw & \qw
    & \qw & \qw & \qw
    & \qw & \qw & \qw
    & \meter{} \\
  \lstick[wires=4]{$\text{Block }2\; \ket{0}^{\otimes n}$}
      & \gate[wires=4,style={rounded corners,fill=blue!8}]{\texttt{OneHotBlock}}
      & \qw
      & \qw
      & \gate[wires=4,style={rounded corners,fill=orange!20}]{U_M^{(2)}(\beta_1)}
      & \qw
      & \qw
      & \gate[wires=4,style={rounded corners,fill=orange!20}]{U_M^{(2)}(\beta_2)}
      & \qw
      & \qw
      & \gate[wires=4,style={rounded corners,fill=orange!20}]{U_M^{(2)}(\beta_3)}
      & \qw
      & \meter{} \\
  & & \qw
    & \qw & \qw & \qw
    & \qw & \qw & \qw
    & \qw & \qw & \qw
    & \meter{} \\
  & & \qw
    & \qw & \qw & \qw
    & \qw & \qw & \qw
    & \qw & \qw & \qw
    & \meter{} \\
  & & \qw
    & \qw & \qw & \qw
    & \qw & \qw & \qw
    & \qw & \qw & \qw
    & \meter{} \\
  \end{quantikz}}
  \caption{Depth-$p=3$ CE-QAOA for $m=3$ blocks of $n=4$ qubits.
  Each layer applies a global cost $U_C(\gamma_\ell)$ over all $mn$ wires,
  followed by parallel block-local XY mixers $U_M^{(j)}(\beta_\ell)$.}
  \label{fig:full-Blockqaoa}
\end{figure}
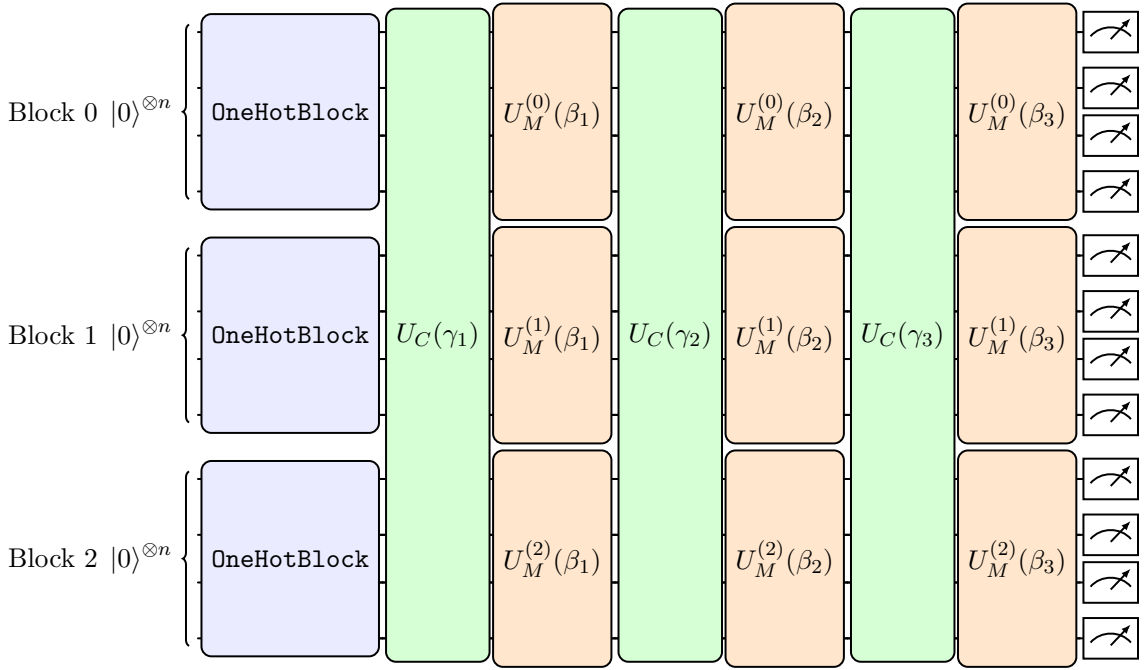

\section{Existence of constant feasibility probability}
\label{sec:invpoly-feasibility}

Several dynamical properties of the normalized block XY mixer proposed in Def. \ref{def:kernel-requirement}  provide good materials for feasibility guarantees at finite depth.   First, it affords a controllable, gapped mixer on the complete graph $K_n$ (Props. \ref{prop:spectral-gap} and \ref{prop:quditization}). These properties lead to ergodicity and global primitivity of a single mixer layer (Lem. \ref{lem:global-primitivity-one-layer}) and ultimately yield a mixer envelope term that prevents the factorization of the probability profile  in Eq.~\ref{eq:factorize} from collapsing to zero. 

\subsection{Single block primitivity}

Contrary to Ref. \cite{onahce} where spectral properties like controllability and approximate universality were established for single blocks within the CE-QAOA circuits, here, we take on the task of establishing the global primitivity of a single mixer layer. Where global in this context refers to the full circuit as opposed to single blocks of qubits(See Fig \ref{fig:full-Blockqaoa}). For the sake of readability, let us recall the key results on the spectral structure, induced interaction graph and controllability of a single block developed in Ref. \cite{onahce}.

\begin{proposition}[Spectral gap of one-block XY mixer]
\label{prop:spectral-gap}
On $\mathcal H_{1}$ the operator $H_{XY}$ acts as the adjacency matrix
$A(K_{n})$ of the complete graph on $n$ vertices and has spectrum
\[
   \operatorname{spec}\bigl(H_{XY}\!\upharpoonright_{\mathcal H_{1}}\bigr)
   = \bigl\{\,n-1,\;\underbrace{-1,\dots,-1}_{n-1\text{ times}}\bigr\}.
\]

Hence the spectral gap is $\Delta(H_{XY}) = n$.  
\(
   \widetilde H_{XY}=H_{XY}/n
\)
has constant gap
\(
   \Delta(\widetilde H_{XY}) = 1.
\)
\end{proposition}

\begin{proposition}[Invariance and quditization of the block–XY mixer]
\label{prop:quditization}
Let
\(
H_{M}^{\mathrm{(b)}}=\sum_{1\le i<j\le n}\bigl(X_i^{(b)}X_j^{(b)}+Y_i^{(b)}Y_j^{(b)}\bigr)
=\sum_{i\neq j}\sigma_{bi}^{-}\sigma_{bj}^{+}
\)
on the \(n\) qubits of a block.
Then:
\begin{enumerate}[leftmargin=1.5em]
\item \(\mathcal H_1\) is invariant under \(H_{M}^{\mathrm{(b)}}\) and \(U^{\mathrm{(b)}}(\beta):=e^{-i\beta H_{M}^{\mathrm{(b)}}}\).
\item In the encoded qudit picture,
\[
V^\dagger\,H_{M}^{\mathrm{(b)}}\,V
\;=\;
\sum_{i\neq j}\ket{i}\!\bra{j}
\;=\; A(K_n),
\]
the adjacency matrix of the complete graph \(K_n\).

Consequently,
\(
U^{\mathrm{(b)}}(\beta)
= V\,e^{-i\beta\,A(K_n)}\,V^\dagger
\)
on \(\mathcal H_1\).
\end{enumerate}
\end{proposition}

\begin{proposition}[Ergodicity of the angle-averaged XY mixer on $\mathcal H_1$]
\label{prop:xy-ergodicity}
Consider a single $n$-qubit one-hot block with one-excitation sector
$\mathcal H_1=\mathrm{span}\{\ket{e_0},\dots,\ket{e_{n-1}}\}$.
Let $H_{M}^{(b)} \!\upharpoonright_{\mathcal H_1}=A(K_n)$ be the restriction of the all-to-all unnormalized XY Hamiltonian to $\mathcal H_1$ (equivalently, the adjacency matrix of the complete graph on $n$ vertices up to an overall scalar).
For $\beta\in\mathbb R$ define $U(\beta):=e^{-i\beta H_{M}^{(b)} }$ and the transition matrix
\[
P_{ij} := \int_{0}^{2\pi} \frac{\mathrm{d}\beta}{2\pi}\,
\left|\langle e_j \,|\, U(\beta) \,|\, e_i \rangle\right|^{2},
\qquad 1 \le i,j \le n.
\]

Then:
\begin{enumerate}
    \item $P$ is \emph{primitive} (all entries are strictly positive), hence the associated Markov chain is \emph{ergodic} (irreducible and aperiodic).
    \item $P$ is \emph{doubly stochastic}, and its unique stationary distribution is the uniform distribution $\pi^\star=(1/n,\dots,1/n)$.
    \item Explicitly,
    \[
    P_{ii}\;=\;1-\frac{2}{n}+\frac{2}{n^2},
    \qquad
    P_{ij}\;=\;\frac{2}{n^2}\quad (j\neq i),
    \]
    so that $P^t\to \mathbf 1\,\pi^{\star\!\top}$ as $t\to\infty$.
\end{enumerate}
These conclusions are invariant under any nonzero rescaling $H_{M}^{(b)} \mapsto c\,H_{M}^{(b)} $, $c\in\mathbb R\setminus\{0\}$; and in particular for the gapped mixer \(   \widetilde H_{XY}=H_{XY}/n\)
from Prop. \ref{prop:spectral-gap}.
\end{proposition}


We now ask when a depth $-p$ CE--QAOA circuit
\[
  \ket{\psi_p(\vec\gamma,\vec\beta)}
  \;=\;
  \Biggl(\prod_{\ell=1}^p e^{-i\beta_\ell H_M} e^{-i\gamma_\ell H_C}\Biggr)\ket{s_0}
\]
induces a primitive Markov kernel on the encoded space when one
``classicalises'' by taking entrywise moduli squared in the
computational basis.  Since $H_C$ is diagonal, the cost unitary
$e^{-i\gamma_\ell H_C}$ contributes only phases and does not affect
transition probabilities. Thus primitivity is entirely controlled by the mixer angles $\beta_\ell$. 

For one block (one-excitation sector $\mathcal H_1$), define the unistochastic transition kernel
\begin{equation}
  M^{(1)}(\beta)_{ij}
  \;:=\;
  \left|\langle e_j \,|\, e^{-i\beta H_M}\,|\, e_i \rangle\right|^{2},
  \qquad 1 \le i,j \le n,
\end{equation}
where $H_M\!\upharpoonright_{\mathcal H_1}=A(K_n)$.
Using the spectral decomposition of $A(K_n)$, one obtains the explicit form
\begin{equation}
\label{eq:M1beta-explicit}
  M^{(1)}(\beta)_{ij}
  \;=\;
  \begin{cases}
    1 - \dfrac{4(n-1)}{n^2}\,\sin^2\!\bigl(\tfrac{n\beta}{2}\bigr),
      & i=j,\\[6pt]
    \dfrac{4}{n^2}\,\sin^2\!\bigl(\tfrac{n\beta}{2}\bigr),
      & i\neq j.
  \end{cases}
\end{equation}
Since $M^{(1)}(\beta)$ is the entrywise modulus-square of a unitary matrix,
it is \emph{doubly stochastic} for every $\beta$.

\paragraph{Angle-averaged kernel.}
Define the $\beta$-averaged kernel as in Prop. \ref{prop:xy-ergodicity}
\begin{equation}
\label{eq:Mbar1-def}
  \overline M^{(1)}_{ij}
  \;:=\;
  \int_{0}^{2\pi} \frac{d\beta}{2\pi}\,M^{(1)}(\beta)_{ij}.
\end{equation}
Using $\int_0^{2\pi}\frac{d\beta}{2\pi}\sin^2(\tfrac{n\beta}{2})=\tfrac12$, we get the constant entries
\begin{equation}
\label{eq:Mbar1-explicit}
  \overline M^{(1)}_{ii}\;=\;1-\frac{2}{n}+\frac{2}{n^2},
  \qquad
  \overline M^{(1)}_{ij}\;=\;\frac{2}{n^2}\quad (i\neq j),
\end{equation}
which matches Proposition~\ref{prop:xy-ergodicity} (after the notational identification
$\overline M^{(1)}\equiv P$ there).

\begin{lemma}[Single--block primitivity]
\label{lem:single-block-primitivity}
For $n\ge 2$, the kernel $M^{(1)}(\beta)$ in \eqref{eq:M1beta-explicit}
has strictly positive entries if and only if
\[
  \sin^2\!\Bigl(\tfrac{n\beta}{2}\Bigr) > 0
  \quad\Longleftrightarrow\quad
  \beta \notin \frac{2\pi}{n}\,\mathbb{Z}.
\]
In this case $M^{(1)}(\beta)$ is primitive. When
$\beta \in \tfrac{2\pi}{n}\mathbb{Z}$, $M^{(1)}(\beta)=I_n$, hence not primitive.
\end{lemma}

\begin{proof}
From \eqref{eq:M1beta-explicit}, $M^{(1)}(\beta)_{ii}>0$ for all $\beta$ and all $i$.
For $i\neq j$,
\[
M^{(1)}(\beta)_{ij}=\frac{4}{n^2}\sin^2\!\Bigl(\frac{n\beta}{2}\Bigr)>0
\quad\Longleftrightarrow\quad
\sin^2\!\Bigl(\frac{n\beta}{2}\Bigr)>0
\quad\Longleftrightarrow\quad
\beta\notin\frac{2\pi}{n}\mathbb Z.
\]
Thus, for $\beta\notin\frac{2\pi}{n}\mathbb Z$, all entries are strictly positive, so
$M^{(1)}(\beta)$ is primitive. If $\beta\in\frac{2\pi}{n}\mathbb Z$, then
$\sin^2(\tfrac{n\beta}{2})=0$ and $M^{(1)}(\beta)=I_n$,
which has zero off-diagonals and is not primitive.
\end{proof}

\subsection{Regimes of global primitivity}
\label{sec:pf-primitivity}

Consider $m$ blocks evolving in parallel with the same mixer angle $\beta$, the global kernel
on the encoded basis $\{\ket{\mathbf{i}}=\ket{i_1,\dots,i_m}\}$ is
\begin{equation}
\label{eq:Mm-def}
  M^{(m)}(\beta)
  \;:=\;
  \bigl(M^{(1)}(\beta)\bigr)^{\otimes m},
  \qquad
  M^{(m)}(\beta)_{\mathbf{i},\mathbf{j}}
  \;=\;
  \prod_{b=1}^{m} M^{(1)}(\beta)_{i_b j_b}.
\end{equation}
Likewise, the averaged global kernel is $\overline M^{(m)} := (\overline M^{(1)})^{\otimes m}$.

\begin{lemma}[Global primitivity for one mixer layer]
\label{lem:global-primitivity-one-layer}
If $\beta\notin \frac{2\pi}{n}\mathbb{Z}$, then
$M^{(m)}(\beta)=\bigl(M^{(1)}(\beta)\bigr)^{\otimes m}$ has strictly
positive entries and is primitive. Consequently, the associated Markov chain
on $[n]^m$ is irreducible and aperiodic.
\end{lemma}

\begin{proof}
By Lemma~\ref{lem:single-block-primitivity}, $M^{(1)}(\beta)$ has strictly positive
entries whenever $\beta\notin\frac{2\pi}{n}\mathbb Z$. The tensor product of
matrices with strictly positive entries again has strictly positive entries, hence
$M^{(m)}(\beta)$ is strictly positive and therefore primitive.
\end{proof}

\begin{corollary}[Perron--Frobenius for the global mixer]
\label{cor:PF-global-mixer}
Assume $\beta\notin \frac{2\pi}{n}\mathbb{Z}$. Then $M^{(m)}(\beta)$ is primitive and
doubly stochastic, hence it has a unique stationary distribution, the uniform vector \cite{Seneta2006}
\[
  \pi^{(m)} \;=\; \frac{1}{n^m}(1,\dots,1).
\]
Moreover, the eigenvalue $\lambda_1=1$ is simple, all other eigenvalues satisfy
$|\lambda_k|<1$, and for any $\mathbf{i}\in[n]^m$,
\[
  \bigl(M^{(m)}(\beta)\bigr)^t \,\delta_{\mathbf{i}}
  \;\xrightarrow[t\to\infty]{}\;
  \pi^{(m)},
\]
with exponential convergence governed by $1-\max_{k\ge 2}|\lambda_k|>0$.
\end{corollary}

\paragraph{Non-resonant parameter choices and robustness}
\label{sec:clarify-params}
The explicit form of the single-block kernel $M^{(1)}(\beta)$ shows that strict positivity (and hence primitivity) fails only at a measure-zero resonance set $\beta\in\frac{2\pi}{n}\mathbb Z$, where $M^{(1)}(\beta)=I_n$.
Away from this set, primitivity holds in a single layer, and products of such kernels remain strictly positive. Consequently, mild randomization of $\beta$ (or small experimental angle noise) generically preserves primitivity, supporting robustness of the mixer envelope $W_p$ as a baseline exploration mechanism. For constrained quantum optimization on the encoded space, the above
results show that, for any choice of mixer angles
$\beta_\ell \notin \frac{2\pi}{n}\mathbb{Z}$ (for at least one layer),
the CE--QAOA mixer is \emph{ergodic} on the encoded manifold and its classicalised dynamics has no non-trivial invariant subsets \cite{Seneta2006,tsvelikhovskiy2024symmetries}.  In particular, any persistent bias in the measured bitstring distribution towards low-cost feasible solutions cannot be attributed to a lack of exploration by the mixer; it must
arise from coherent interference induced by the cost Hamiltonian and the specific choice of angles $(\vec\gamma,\vec\beta)$.  This justifies interpreting the uniform distribution $\pi^{(m)}$ as a natural ``null model'' (or design-based baseline) for the encoded space, and viewing deviations from this baseline as genuine algorithmic structure\cite{onahfund}.

\subsection{Feasibility from finite level transition}
\label{sec:feasibility-from-finite-level-transition}
\medskip

\begin{definition}[Penalty level-set states]
\label{def:levelset-states}
Let $H_{\mathrm{pen}}$ satisfy Def.~\ref{def:kernel-requirement}(a,b), with integer spectrum
$\mathrm{spec}(H_{\mathrm{pen}})\subseteq\{0,1,\dots,t_{\max}\}$ on $\OH$, $t_{\max}=\poly n$.
For each level $t\in\{0,1,\dots,t_{\max}\}$ define the level set
\(
  L_t:=\{z\in[n]^m:\ H_{\mathrm{pen}}(z)=t\}
\)
and the normalized uniform level-set vector
\[
  \ket{L_t}
  \;:=\;
  \frac{1}{\sqrt{|L_t|}}\sum_{z\in L_t}\ket{z}
  \qquad (\text{if }|L_t|>0).
\]
Let $\mathcal V:=\{t:\ |L_t|>0\}$ and define the \emph{level-set subspace}
\[
  \mathcal K \;:=\; \mathrm{span}\{\ket{L_t}:\ t\in\mathcal V\}\subseteq \OH,
  \qquad d:=\dim\mathcal K = |\mathcal V|\le t_{\max}+1 = \poly n.
\]
\end{definition}

\paragraph{A controllability criterion.}
Write the restrictions $A:=H_{\mathrm{pen}}|_{\mathcal K}$ and $B:=\widetilde H_M|_{\mathcal K}$.
In the $\{\ket{L_t}\}_{t\in\mathcal V}$ basis, $A$ is diagonal with entries $t$.
The matrix $B$ is real symmetric (since $\widetilde H_M$ is), and has off-diagonal entries whenever
the mixer connects basis strings across different penalty levels.

\begin{definition}[Level-transition graph]
\label{def:level-graph}
Define the undirected \emph{level-transition graph}
\[
\Gamma \;=\; (\mathcal V,\mathcal E)
\]
with vertex set \(\mathcal V\subseteq \{0,1,\dots,t_{\max}\}\) given by the active
penalty levels
\[
\mathcal V \;:=\; \{\,t:\ |L_t|>0\,\},
\]
and edge set
\[
\mathcal E \;:=\; \bigl\{\,\{t,t'\}\subseteq \mathcal V:\ \langle L_{t'}|B|L_t\rangle \neq 0\,\bigr\}.
\]
\end{definition}

\medskip

For the canonical column-collision penalty used in assignment/permutation encodings (and hence TSP/QAP/CVRP one-hot cores), one can verify that $\mathsf \Gamma$ is connected implying that from any infeasible string, a single-block relabeling can reduce the column-penalty value, so there exist mixer edges from any level toward lower levels, ultimately reaching $t=0$. The next three Lemmas make this structure explicit, culminating in a constant feasibility probability guarantee at finite depth.

\begin{lemma}[Strict penalty descent by a single block relabeling (permutation case $m=n$)]
\label{lem:descent}
Consider the permutation/assignment constraint with $m=n$ blocks and column-counts
$N_k(z):=\#\{b:\ z_b=k\}$.
Let
\[
  H_{\mathrm{pen}}(z) \;=\; \sum_{k=0}^{n-1} (N_k(z)-1)^2,
\]
so $H_{\mathrm{pen}}(z)=0$ iff all $N_k(z)=1$ (a permutation).
If $z$ is infeasible, then there exist symbols $a,b$ with $N_a(z)\ge 2$ and $N_b(z)=0$.
Relabeling one block currently equal to $a$ into $b$ produces $z'$ with
\[
  H_{\mathrm{pen}}(z') \;\le\; H_{\mathrm{pen}}(z) - 2.
\]
\end{lemma}

\begin{proof}
Infeasible with $\sum_k N_k(z)=n$ implies some $N_a\ge 2$ and some $N_b=0$. Only the $a$ and $b$ terms change under moving one unit from $a$ to $b$: $N_a\mapsto N_a-1$, $N_b\mapsto N_b+1$. Using $f(N)=(N-1)^2$,
\[
  f(N_a-1)+f(N_b+1)-f(N_a)-f(N_b)
  \;=\;
  2(N_b-N_a+1)
  \;\le\;
  2(0-2+1)=-2,
\]
so $H_{\mathrm{pen}}$ decreases by at least $2$.
\end{proof}
\begin{lemma}[Connectivity of the level-transition graph $\Gamma$]
\label{lem:connected-graph}
Under the same penalty (and with the block-XY mixer),
the level-transition graph $\Gamma=(\mathcal V,\mathcal E)$ on $\mathcal V$ is connected.
\end{lemma}

\begin{proof}
Lemma~\ref{lem:descent} gives, for any $t>0$ with $L_t\neq\emptyset$, the existence of basis strings
$z\in L_t$ and $z'\in L_{t'}$ with $t'<t$ that differ by a single block relabeling. The block-XY mixer has nonzero matrix elements between such $z$ and $z'$, hence $\langle L_{t'}|B|L_t\rangle\neq 0$ for some $t'<t$, i.e.\ $\{t,t'\}\in\mathcal E$. Iterating the descent reaches $t=0$, establishing connectivity.
\end{proof}

\begin{lemma}[Invariant-sector controllability (hypothesis)]
\label{lem:connected-implies-su}
Assume Def.~\ref{def:kernel-requirement}(b,c) and let $G:=S_m\times S_n$ act on $\OH$ by
block permutations and global symbol relabelings. Define the invariant symmetry sector
\[
\mathcal K_{\mathrm{inv}}:=\mathrm{Fix}(G)
=\{\ket{\psi}\in\OH:\ U_g\ket{\psi}=\ket{\psi}\ \forall g\in G\},
\qquad d:=\dim\mathcal K_{\mathrm{inv}}.
\]
Let
\[
A:=H_{\mathrm{pen}}\!\upharpoonright_{\mathcal K_{\mathrm{inv}}},
\qquad
B:=\widetilde H_M\!\upharpoonright_{\mathcal K_{\mathrm{inv}}}.
\]
\emph{Hypothesis:} the real Lie algebra generated by $\{iA,iB\}$ is $\mathfrak u(d)$.
Equivalently, the dynamical group generated by $\exp(-i\gamma A)$ and $\exp(-i\beta B)$
is dense in $U(d)$ on $\mathcal K_{\mathrm{inv}}$.
\end{lemma}

\begin{proof}
This lemma records the controllability hypothesis on the invariant symmetry sector
$\mathcal K_{\mathrm{inv}}$ which is the relevant dynamical subspace for the feasibility stage.
\end{proof}

\medskip
We are now ready to prove the existence of a finite probability guarantee in Constraint-Enhanced Quantum optimization.

\begin{theorem}[Finite-depth feasibility on the invariant sector]
\label{thm:finite-depth-feasible}
Assume the CE--QAOA kernel conditions (Def.~\ref{def:kernel-requirement}) and the
permutation/assignment penalty with $m=n$,
\[
H_{\mathrm{pen}}(z)=\sum_{k=0}^{n-1}\bigl(N_k(z)-1\bigr)^2,
\qquad N_k(z):=\#\{b:\ z_b=k\}.
\]
Let $G:=S_m\times S_n$ and $\mathcal K_{\mathrm{inv}}:=\mathrm{Fix}(G)\subseteq\OH$ with $d:=\dim\mathcal K_{\mathrm{inv}}$.
Define
\[
A:=H_{\mathrm{pen}}\!\upharpoonright_{\mathcal K_{\mathrm{inv}}},
\qquad
B:=\widetilde H_M\!\upharpoonright_{\mathcal K_{\mathrm{inv}}}.
\]
Assume the controllability hypothesis of Lemma~\ref{lem:connected-implies-su}, i.e.\
$\mathrm{Lie}\{iA,iB\}=\mathfrak u(d)$.

Let $\Pi_0:=\sum_{z\in L_0}\ket{z}\!\bra{z}$ be the projector onto the feasible level set.
For a depth-$p$ alternating product on the feasibility stage,
\[
U_p(\boldsymbol\gamma,\boldsymbol\beta)
:=\prod_{j=1}^{p} e^{-i\beta_j B}\,e^{-i\gamma_j A},
\qquad
\ket{\psi_p(\boldsymbol\gamma,\boldsymbol\beta)}:=U_p(\boldsymbol\gamma,\boldsymbol\beta)\ket{s_0},
\]
define the feasibility probability
\[
\pi_{\mathsf F}(p;\boldsymbol\gamma,\boldsymbol\beta)
:=\bra{\psi_p(\boldsymbol\gamma,\boldsymbol\beta)}\,\Pi_0\,\ket{\psi_p(\boldsymbol\gamma,\boldsymbol\beta)}.
\]
Then for every $\eta\in(0,1)$ there exist a finite depth $p<\infty$ and angles
$(\boldsymbol\gamma^\star,\boldsymbol\beta^\star)$ such that
\[
\pi_{\mathsf F}(p;\boldsymbol\gamma^\star,\boldsymbol\beta^\star)\ \ge\ 1-\eta.
\]
In particular, there exists some finite depth and angles such that $\pi_{\mathsf F}\ge \tfrac12$.
\end{theorem}

\begin{proof}

\textbf{Step 1 (reachability on $\mathcal K_{\mathrm{inv}}$).}
By the controllability hypothesis of Lemma~\ref{lem:connected-implies-su},
$\mathrm{Lie}\{iA,iB\}=\mathfrak u(d)$ on $\mathcal K_{\mathrm{inv}}$.
Hence the group generated by $\exp(-i\gamma A)$ and $\exp(-i\beta B)$ is dense in $U(d)$ on
$\mathcal K_{\mathrm{inv}}$ (see, e.g., \cite{DAlessandro2007,Altafini2002}).
Therefore, for any target unit vector $\ket{\phi}\in\mathcal K_{\mathrm{inv}}$ and any $\varepsilon>0$,
there exist a finite depth $p$ and angles $(\boldsymbol\gamma,\boldsymbol\beta)$ such that
\begin{equation}
\label{eq:approx-reach}
\bigl\|\,\ket{\psi_p(\boldsymbol\gamma,\boldsymbol\beta)}-\ket{\phi}\,\bigr\|_2 \ \le\ \varepsilon.
\end{equation}
We apply this with $\ket{\phi}=\ket{L_0}$, noting that $\ket{s_0}\in\mathcal K_{\mathrm{inv}}$ and
$\ket{L_0}\in\mathcal K_{\mathrm{inv}}$ since both are invariant under block permutations $S_m$
and global symbol relabelings $S_n$.

\textbf{Step 2 (convert proximity to a feasibility lower bound).}
Feasibility probability is
\[
\pi_{\mathsf F}
=\|\Pi_0\ket{\psi_p}\|_2^2.
\]
Since $\ket{L_0}\in\mathrm{Ran}(\Pi_0)$ and $\|\ket{L_0}\|_2=1$, we have
\begin{equation}
\label{eq:feas-lb-by-overlap}
\pi_{\mathsf F}
=\|\Pi_0\ket{\psi_p}\|_2^2
\ \ge\
\bigl|\langle L_0\mid \psi_p\rangle\bigr|^2.
\end{equation}
Using
\[
\|\ket{\psi_p}-\ket{L_0}\|_2^2
=2-2\,\mathrm{Re}\langle L_0\mid \psi_p\rangle,
\]
the reachability bound \eqref{eq:approx-reach} implies
\[
\mathrm{Re}\langle L_0\mid \psi_p\rangle \ge 1-\varepsilon^2/2,
\]
and therefore
\begin{equation}
\label{eq:overlap-lb}
\bigl|\langle L_0\mid \psi_p\rangle\bigr|^2
\ \ge\
\bigl(1-\varepsilon^2/2\bigr)^2.
\end{equation}
Combining \eqref{eq:feas-lb-by-overlap} and \eqref{eq:overlap-lb} yields
\[
\pi_{\mathsf F}\ \ge\ \bigl(1-\varepsilon^2/2\bigr)^2.
\]

\textbf{Step 3 (choose $\varepsilon$ for the desired target).}
Given any $\eta\in(0,1)$, choose $\varepsilon$ so that
\(
(1-\varepsilon^2/2)^2\ge 1-\eta,
\)
e.g.
\[
\varepsilon
\ \le\
\sqrt{\,2\bigl(1-\sqrt{1-\eta}\bigr)\,}.
\]
Then by \eqref{eq:approx-reach} there exist finite $p$ and angles achieving this $\varepsilon$,
hence $\pi_{\mathsf F}\ge 1-\eta$.

For the concrete constant claim $\pi_{\mathsf F}\ge \tfrac12$, it suffices to take $\varepsilon=\tfrac12$, which gives
\(
\pi_{\mathsf F}\ge (1-\tfrac{1}{8})^2=\tfrac{49}{64}>\tfrac12.
\)
\end{proof}

Theorem~\ref{thm:finite-depth-feasible} is an \emph{existence} guarantee at finite depth by a CE--QAOA parameter choice exploiting the kernel structure.
A quantitative finite-depth lower bound follows from the Fej\'er factorization law (Eq.~\ref{eq:factorize}) when we use \emph{penalty-only phases} $H_C\equiv H_{\mathrm{pen}}$ (App.~\ref{sec:fejer-feasibility}).


\section{Positive Trigonometric Filters from Prepared Quantum States}
\label{sec:fejer-proof}

\subsection{Fejér Kernel and Its Basic Inequalities}
\label{app:fejer}

\begin{definition}[Fejér kernel {\cite{Katznelson2004}}]
For an integer \(p\ge 0\), the Fejér kernel \(F_p:\mathbb R\to\mathbb R_{\ge 0}\) is
\begin{equation}
  F_p(\theta)
  \;:=\;
  \frac{1}{p+1}\left(\frac{\sin\!\bigl(\frac{(p+1)\theta}{2}\bigr)}{\sin\!\bigl(\frac{\theta}{2}\bigr)}\right)^2
  \;=\;
  \sum_{m=-p}^{p} a_m^{(p)} e^{i m \theta},
  \qquad
  a_m^{(p)} \;=\; \frac{p+1-|m|}{p+1}\;\;\;(\ge 0).
  \label{eq:fejer-def}
\end{equation}
\end{definition}

The following Lemma collects the key properties of interest.
\begin{lemma}[Fejér facts {\cite{Katznelson2004}}]\label{lem:fejer-facts}
For all \(p\ge 0\) and \(\theta\in\mathbb R\):
\begin{enumerate}[label=(\alph*),leftmargin=2.2em]
\item (Positivity) \(F_p(\theta)\ge 0\).
\item (Normalization) \(\displaystyle \frac{1}{2\pi}\int_{-\pi}^{\pi} F_p(\theta)\,d\theta = 1\).
\item (Peak value) \(F_p(0)=p+1\).
\end{enumerate}
\end{lemma}

\begin{proof}[Sketch]
See, e.g., \cite[§1.2--1.3]{SteinShakarchi2003Fourier} or \cite[§II.1]{Katznelson2004} for (a)-(c).
\end{proof}

\subsection{CE–QAOA and Positive Trigonometric Filters on \texorpdfstring{$e^{-i\gamma H_C}$}{exp(-i gamma H	extunderscore C)}}
\label{sec:fejer}

\paragraph{Lattice-normalized phases and gaps.}
Throughout we assume that after a known global rescaling of the cost Hamiltonian, the eigenphases of the cost unitary $U_C(\gamma)=e^{-i\gamma H_C}$ lie on an integer lattice and admit a \emph{gap} around the optimal phase. A Riemann--Lebesgue averaging beyond exact lattice normalization is discussed in Sec. \ref{subsec:rl-averaging}. Concretely, we assume that there exists a scale $\Lambda_C>0$
such that, under the redefinition
\[
H_C \longleftarrow H_C/\Lambda_C,
\]
the rescaled Hamiltonian has spectrum in $\mathbb Z$
(or can be made so up to an application-dependent discretization). We then restrict $\gamma$ to the corresponding lattice so that $\theta(z):=\gamma\,E_C(z)$ is well-defined modulo $2\pi$ with controlled resolution. 

Let \(\Omega^\star\subseteq\mathcal F\subseteq\mathcal H_{\mathrm{OH}}\)
denote the set of optimal feasible basis strings, all having wrapped
phase \(\theta^\star\). We assume that the target phase is separated
from every nontarget phase on the encoded space:
\begin{equation}
\label{eq:phase-gap-again}
\operatorname{dist}_{\mathbb T}
\bigl(\theta(z),\theta^\star\bigr)
\ge \delta,
\qquad
\forall z\in
\mathcal H_{\mathrm{OH}}\setminus\Omega^\star .
\end{equation}
where
\(
\operatorname{dist}_{\mathbb T}(\phi,\varphi)
:=
\min_{k\in\mathbb Z}
|\phi-\varphi+2\pi k|
\in[0,\pi].
\)

\begin{remark}[Unknown optimal phase and phase-grid refinement]
\label{rem:unknown-optimal-phase}
The exact-optimum bound is stated for a filter centered at the optimal
wrapped phase $\theta^\star$. When $\theta^\star$ is unknown, one may
instead consider a finite grid of candidate phase centers
\[
\Theta_h:=\{\vartheta_0,\ldots,\vartheta_{N_h-1}\}
\subset[0,2\pi),
\]
with maximal circular spacing $h$. There then exists a grid point
$\vartheta_h^\star\in\Theta_h$ satisfying
\[
\operatorname{dist}_{\mathbb T}
(\vartheta_h^\star,\theta^\star)
\le \frac{h}{2}.
\]
The Fej\'er weight assigned to an optimal configuration by the
corresponding filter is therefore \(
F_p(\theta^\star-\vartheta_h^\star),
\)
which converges to the peak value $F_p(0)=p+1$ as $h\to0$. More
quantitatively, if \(
\frac{h}{2}\le \frac{c}{p+1}
\)
for some fixed $c\in(0,\pi)$, the main-lobe estimate gives \(
F_p(\theta^\star-\vartheta_h^\star)
\ge \kappa_c(p+1),
\)
where $\kappa_c>0$ depends only on $c$. Thus an unknown optimum may be approached through a progressively
refined phase grid, or through adaptive phase anchors obtained from
incumbents and relaxation bounds. 

\end{remark}

 As the grid spacing decreases, the phase-resolution error vanishes and the selected window approaches the optimal phase sector. Under this spectral/phase regularity condition, the Fej\'er filter exposed by lattice-normalized angles yields dimension-free finite-depth and finite-shot success bounds.

Recall that
\[
\OH \;=\; (\mathcal H_1)^{\otimes m},
\qquad
\mathcal H_1 \;=\; \mathrm{span}\{\ket{e_1},\dots,\ket{e_n}\}
\quad\text{(one excitation per block)}.
\]
The mixer unitary is given as,
      \[
      U_M(\beta) \;=\; \bigotimes_{b=1}^{m} e^{-i\beta\,\widetilde H_{XY}^{(b)}}.
      \]
     
      The initial state is the
      uniform one–hot product
      \[
      \ket{s_0}\;=\;\ket{s_{\mathrm{blk}}}^{\otimes m},
      \qquad
      \ket{s_{\mathrm{blk}}}\;=\;\frac{1}{\sqrt n}\sum_{j=1}^n \ket{e_j}
      \quad\text{(a $W_n$ state per block)}.
      \]

Let \(H_C\) be a diagonal cost Hamiltonian on a sector \(D\subseteq (\mathbb C^n)^{\otimes m}\)
with the lattice normalized spectrum \(\{E(z): z\in\OH\}\) on the computational basis \(\{\ket{z}\}_{z\in\OH}\).
Let \(U_C(\gamma)=e^{-i\gamma H_C}\). We start from the \(p\)-layer CE--QAOA state
\begin{equation}
  \label{eq:state-def-again}
  |\psi_p\rangle
  \;=\;
  \Bigl(\prod_{r=1}^{p} U_M(\beta_r)\,U_C(\gamma_r)\Bigr)\,|s_0\rangle,
  \qquad
  U_C(\gamma)=e^{-i\gamma H_C},
\end{equation}
and write amplitudes in the computational basis \(\{|z\rangle\}_{z\in\OH}\), with
\(H_C|z\rangle = E(z)\,|z\rangle\).
Fix any optimal phase \(\theta^\star \in \Theta^\star := \{\gamma E^\star \ (\!\!\!\!\!\!\mod 2\pi)\}\),
where \(E^\star=\min_z E(z)\). Throughout we choose a \emph{harmonic} schedule
\begin{equation}
  \label{eq:harmonic-sched}
  \gamma_r \;=\; r\,\gamma,\qquad r=1,\dots,p,
\end{equation}
so that \(\{e^{-i\,r\,\gamma H_C}\}_{r=1}^p\) generates the first \(p\) Fourier harmonics of the cost phases.  The Fej\'er kernel is exposed as follows:
\paragraph{Step 1: Cost-basis dephasing/twirling.}

Let \(\mathcal T\) be the cost-basis dephasing channel
\[
  \mathcal T(\rho)
  \;=\;\int_{0}^{2\pi}\!\frac{d\phi}{2\pi}\,e^{-i\phi H_C}\rho\,e^{+i\phi H_C},
\]
i.e.\ the unitary \emph{twirl} over the commuting one-parameter group
\(\{e^{-i\phi H_C}\}_{\phi\in[0,2\pi)}\). 
Under the lattice-normalization assumption, this twirl is a conditional expectation onto the
$H_C$-eigenspace blocks (off-diagonal terms with distinct eigenvalues average to zero over $[0,2\pi)$) \cite{Watrous2018}. Operationally, one may realize
the same averaged channel by sampling \(\phi\) uniformly at random in each run and
forgetting \(\phi\) (a standard “twirling” construction) \cite{BennettEtAl1996Purification}.

\medskip
\noindent\emph{Note.} We use \(\mathcal T\) as a \emph{classicalizing baseline model}
that suppresses interference terms, thereby exposing a positive (nonnegative-coefficient)
trigonometric filtering mechanism. We do \emph{not} claim the resulting surrogate evolution
is identical to the fully coherent CE--QAOA circuit; rather, it provides a rigorous,
interference-free reference distribution against which coherent effects can only
redistribute weight by additional interference \cite{Watrous2018}.

\paragraph{Step 2: The mixer induces a classical Markov kernel defining \(W_p\).}
Because \(\mathcal T\) removes off-diagonal coherences immediately after each layer,
the diagonal update across a single layer becomes \emph{classical}:
\begin{align}
  v^{(r)}(z)
  &= \sum_{y\in\OH} \underbrace{\bigl|\langle z|U_M(\beta_r)|y\rangle\bigr|^2}_{=:M_{\beta_r}(z|y)}
     \cdot \underbrace{\langle y|U_C(\gamma_r)\,\rho_{r-1}\,U_C(\gamma_r)^\dagger|y\rangle}_{=\,v^{(r-1)}(y)}
  \notag\\
  &= \sum_{y\in\OH} M_{\beta_r}(z|y)\,v^{(r-1)}(y),
\end{align}
since \(U_C(\gamma_r)\) is diagonal and does not change the diagonal of \(\rho_{r-1}\).
Therefore, with \(M_{\beta}\) the stochastic matrix
\[
  M_{\beta}(z|y)\;:=\;\bigl|\langle z|U_M(\beta)|y\rangle\bigr|^2,
\]
we get
\begin{equation}
  \label{eq:Wp-def}
  v^{(p)} \;=\; M_{\beta_p}\cdots M_{\beta_1}\,v^{(0)},
  \qquad
  v^{(0)}(z)=|\langle z|s_0\rangle|^2.
\end{equation}
We \emph{define the mixer envelope} \(W_p(\cdot;\boldsymbol\beta)\) as precisely this dephased diagonal:

\[
W_p(z;\boldsymbol\beta)
=
\sum_{y_0,\ldots,y_{p-1}\in\OH}
\left(
\prod_{r=1}^{p}
M_{\beta_r}(y_r\mid y_{r-1})
\right)
v^{(0)}(y_0),
\qquad y_p=z.
\]

with \(y_0=z_0\) understood. On our one–hot product space \(\OH=[n]^m\),
the block–local XY mixer factorizes across blocks \(b=1,\dots,m\):
\begin{equation}
  M_{\beta}(z|y)
  \;=\;
  \prod_{b=1}^{m}\ \Bigl|\bigl\langle j_b\big|e^{-i\beta A(K_n)}\big|k_b\bigr\rangle\Bigr|^2,
  \quad
  y=(k_1,\dots,k_m),\ z=(j_1,\dots,j_m),
\end{equation}
using Proposition~\ref{prop:quditization}. Thus \(W_p\) is a \emph{blockwise product} of single-block XY transition kernels, iterated \(p\) times and applied to the initial blockwise-uniform \(v^{(0)}\).

\paragraph{Digression on ergodicity of the mixer-induced Markov chain.}
After inserting the cost-basis dephasing channel, the diagonal evolves by a
stochastic kernel
\[
  v^{(r)}(z) \;=\; \sum_{y\in\OH} M_{\beta_r}(z|y)\,v^{(r-1)}(y),
  \qquad
  M_{\beta}(z|y)\ :=\ \bigl|\langle z|U_M(\beta)|y\rangle\bigr|^2,
\]
and \(W_p(\cdot;\boldsymbol\beta) = M_{\beta_p}\cdots M_{\beta_1}\,v^{(0)}\).
Since \(M_\beta\) is the entrywise modulus-square of a unitary, it is
\emph{unistochastic} and therefore \emph{doubly stochastic}
(\(\sum_z M_\beta(z|y)=\sum_y M_\beta(z|y)=1\)). Consequently, standard finite-state Markov-chain theory applies
\cite{Seneta2006,LevinPeresWilmer2017}.
(cf.\ continuous-time walk on graphs~\cite{ChildsReview}).

\paragraph{Step 3: Fej\'er weighting via a centered Dirichlet filter.}
To obtain an explicit positive trigonometric filter on the cost phases,
we consider the following filtered variant of the protocol. Given any
prepared, possibly classicalized, envelope, define the centered
normalized Dirichlet filter
\begin{equation}
D_p^\star(H_C)
:=
\frac{1}{\sqrt{p+1}}
\sum_{r=0}^{p}
e^{-ir(\gamma H_C-\theta^\star I)},
\qquad
\theta(z):=\gamma E(z)\pmod{2\pi}.
\label{eq:centered-dirichlet-filter}
\end{equation}
For any computational-basis eigenstate
$H_C\ket{z}=E(z)\ket{z}$,
\begin{equation}
D_p^\star(H_C)\ket{z}
=
\frac{1}{\sqrt{p+1}}
\sum_{r=0}^{p}
e^{-ir(\theta(z)-\theta^\star)}
\ket{z}.
\label{eq:centered-dirichlet-action}
\end{equation}
Consequently,
\begin{equation}
\begin{aligned}
\bra{z}
D_p^\star(H_C)^\dagger
D_p^\star(H_C)
\ket{z}
&=
\frac{1}{p+1}
\left|
\sum_{r=0}^{p}
e^{-ir(\theta(z)-\theta^\star)}
\right|^2                                                   =
F_p\!\left(\theta(z)-\theta^\star\right).
\end{aligned}
\label{eq:centered-fejer-weight}
\end{equation}

\paragraph{Step 4: Factorizing and extracting the \(\theta\)-dependence via a positive trigonometric filter.}
The Fej\'er kernel \(F_p\) is the canonical nonnegative trigonometric polynomial
associated with Ces\`aro-averaging of Fourier partial sums \cite{SteinShakarchi2003Fourier,Katznelson2004}. To obtain an exact factorization for the filtered protocol, let $\rho_{\mathrm{env}}$ be any state whose diagonal is the envelope $W_p(\cdot;\boldsymbol\beta)$.
Define the (postselected/renormalized) filtered state

\begin{equation}
\rho_{\mathrm{filt}}
:=
\frac{
D_p^\star(H_C)\rho_{\mathrm{env}}D_p^\star(H_C)^\dagger
}{
\Tr[
D_p^\star(H_C)\rho_{\mathrm{env}}D_p^\star(H_C)^\dagger
]
}.
\end{equation}

Then the computational-basis measurement distribution of
$\rho_{\mathrm{filt}}$ factorizes, for every $z\in\OH$, as
\begin{equation}
\Pr_{\mathrm{filt}}[z]
=
\frac{
W_p(z;\boldsymbol\beta)\,
F_p(\theta(z)-\theta^\star)
}{
\sum_{y\in\OH}
W_p(y;\boldsymbol\beta)\,
F_p(\theta(y)-\theta^\star)
}.
\end{equation}
Consequently, allowing the optimum to be nonunique,
\begin{equation}
\Pr_{\mathrm{filt}}[\Omega^\star]
=
\frac{
\displaystyle
\sum_{z\in\Omega^\star}
W_p(z;\boldsymbol\beta)\,
F_p(\theta(z)-\theta^\star)
}{
\displaystyle
\sum_{y\in\OH}
W_p(y;\boldsymbol\beta)\,
F_p(\theta(y)-\theta^\star)
}.
\end{equation}

We interpret the factorized law
\begin{equation}
\label{eq:factorize}
\Pr^{\mathrm{ref}}_p[z]
\;\propto\;
W_p(z;\boldsymbol\beta)\,
F_p(\theta(z)-\theta^\star),
\qquad z\in\OH,
\end{equation}
as a \emph{reference} distribution that isolates the
\emph{positive} filtering contribution.

\begin{remark}[Fej\'er-induced nonuniformity]
Even when the pre-filter envelope \(W_p\) is uniform, the filtered law
\[
\Pr_p^{\mathrm{ref}}[z]
\propto
W_p(z;\boldsymbol\beta)F_p(\theta(z)-\theta^\star)
\]
may not be. The analysis therefore concerns the
joint Fej\'er-reweighted distribution.
\end{remark}


The fully coherent CE--QAOA circuit may deviate from
$\Pr^{\mathrm{ref}}_p$ through additional interference that sharpens,
flattens, or otherwise redistributes probability mass. The reference
model instead isolates the underlying positive-kernel mechanism
\cite{Watrous2018,SteinShakarchi2003Fourier}. Here ``order'' denotes
the degree of the positive trigonometric polynomial governing the
classicalized, cost-dephased diagonal statistics. The cost-basis
dephasing channel $\mathcal T$ is used only as an analytic device:
it replaces coherent interference by a pinched surrogate evolution and
isolates the tractable mixer-envelope distribution
$W_p(\cdot;\boldsymbol\beta)$ generated by the resulting classical
Markov evolution. In particular, this classicalization eliminates the
dependence of the diagonal statistics on the cost angles
$\boldsymbol\gamma$.

The explicit Fej\'er dependence enters through the centered
Dirichlet filtering step
\begin{equation}
D_p^\star(H_C)
=
\frac{1}{\sqrt{p+1}}
\sum_{r=0}^{p}
e^{-ir(\gamma H_C-\theta^\star I)},
\end{equation}
whose squared diagonal action produces
$F_p(\theta-\theta^\star)$. The resulting filtered protocol admits the
exact factorization~\eqref{eq:factorize} and the corresponding
dimension-independent reference-model success bounds. Pinching and
twirling constructions of this kind are standard in quantum information
\cite{Watrous2018}; operationally, the averaged channel can be realized
by randomizing over a unitary family and averaging over, or discarding,
the classical random seed \cite{BennettEtAl1996Purification}.


\section{Analyses of Fej\'er Factorization Law}
\label{sec:base-success-three-regimes}
\subsection{Setup}

Define the total mixer weight on
the optima by
\begin{equation}
  \label{eq:Pi-star-def}
  C_{\beta} \;:=\; \sum_{z\in\Omega^\star} W_p(z;\boldsymbol\beta)
  \;\in\; (0,1],
\end{equation}



\begin{lemma}[Off-peak Fej\'er bound]
\label{lem:fejer-offpeak}
For any $\delta\in(0,\pi]$, define
\begin{equation}
\label{eq:Mp-wrapped}
M_p(\delta)
:=
\max_{\substack{\vartheta\in[0,2\pi)\\
\operatorname{dist}_{\mathbb T}(\vartheta,0)\ge\delta}}
F_p(\vartheta).
\end{equation}
Then
\begin{equation}
\label{eq:Mp-bound}
M_p(\delta)
\le
\frac{1}{(p+1)\sin^2(\delta/2)}
\le
\frac{\pi^2}{(p+1)\delta^2}.
\end{equation}
\end{lemma}

\begin{proof}
For each $\vartheta$, let $\bar\vartheta\in[-\pi,\pi]$ be its
principal representative, so that
\[
|\bar\vartheta|
=
\operatorname{dist}_{\mathbb T}(\vartheta,0).
\]
Using the $2\pi$-periodicity of $F_p$,
\[
F_p(\vartheta)
=
\frac{1}{p+1}
\frac{\sin^2((p+1)\bar\vartheta/2)}
     {\sin^2(\bar\vartheta/2)}.
\]
If $\operatorname{dist}_{\mathbb T}(\vartheta,0)\ge\delta$, then
$|\bar\vartheta|\in[\delta,\pi]$. Hence
\[
\left|\sin\!\left(\frac{\bar\vartheta}{2}\right)\right|
\ge
\sin\!\left(\frac{\delta}{2}\right),
\qquad
\left|\sin\!\left(\frac{(p+1)\bar\vartheta}{2}\right)\right|
\le 1,
\]
which proves the first inequality. The second follows from
$\sin(\delta/2)\ge\delta/\pi$ for $\delta\in[0,\pi]$.
\end{proof}

Since \(F_p(0)=p+1\) for all \(z\in\Omega^\star\) and the off-peak Fej\'er lobe gives
\(F_p(\theta(y)-\theta^\star)\le M_p(\delta)\) for \(y\notin\Omega^\star\) (Lem. \ref{lem:fejer-offpeak}), the denominator of
\eqref{eq:factorize} obeys
\[
  \sum_{y} W_p(y;\boldsymbol\beta)\,F_p(\theta(y)-\theta^\star)
  \;\le\;
  (p+1)\,C_{\beta} + M_p(\delta)\,(1-C_{\beta}).
\]
We prove it in the following Lemma.
\begin{lemma}[Weighted Fej\'er denominator bound]
\label{lem:denom-bound}
Let $W_p(\cdot;\boldsymbol\beta)$ be a probability distribution on $\OH$ (i.e., $W_p(y;\boldsymbol\beta)\ge 0$ and $\sum_{y\in\OH}W_p(y;\boldsymbol\beta)=1$).

Then
\[
  \sum_{y\in\OH} W_p(y;\boldsymbol\beta)\,F_p\!\bigl(\theta(y)-\theta^\star\bigr)
  \;\le\;
  (p+1)\,C_{\beta} \;+\; M_p(\delta)\,\bigl(1-C_{\beta}\bigr).
\]
\end{lemma}

\begin{proof}
Split the sum into optimal and non-optimal strings:
\[
\sum_{y\in\OH} W_p(y)F_p(\theta(y)-\theta^\star)
=
\sum_{z\in\Omega^\star} W_p(z)F_p(\theta(z)-\theta^\star)
+
\sum_{y\notin\Omega^\star} W_p(y)F_p(\theta(y)-\theta^\star),
\]
where we abbreviate $W_p(\cdot;\boldsymbol\beta)$ by $W_p(\cdot)$.
For $z\in\Omega^\star$, we have $\theta(z)\equiv\theta^\star\ (\mathrm{mod}\ 2\pi)$, hence
$F_p(\theta(z)-\theta^\star)=F_p(0)=p+1$, and therefore
\[
\sum_{z\in\Omega^\star} W_p(z)F_p(\theta(z)-\theta^\star)=(p+1)\sum_{z\in\Omega^\star}W_p(z)
=(p+1)C_\beta.
\]
For $y\notin\Omega^\star$, the phase-gap assumption implies
$\operatorname{dist}(\theta(y),\theta^\star)\ge\delta$, hence
$F_p(\theta(y)-\theta^\star)\le M_p(\delta)$ by definition of $M_p(\delta)$.
Using $W_p(y)\ge 0$,
\[
\sum_{y\notin\Omega^\star} W_p(y)F_p(\theta(y)-\theta^\star)
\le
M_p(\delta)\sum_{y\notin\Omega^\star} W_p(y).
\]
Finally, since $W_p$ is a probability distribution,
$\sum_{y\notin\Omega^\star} W_p(y)=1-\sum_{z\in\Omega^\star}W_p(z)=1-C_\beta$.
Combining the bounds gives the claim.
\end{proof}

\begin{figure}[t]
  \centering
  \includegraphics[width=\linewidth]{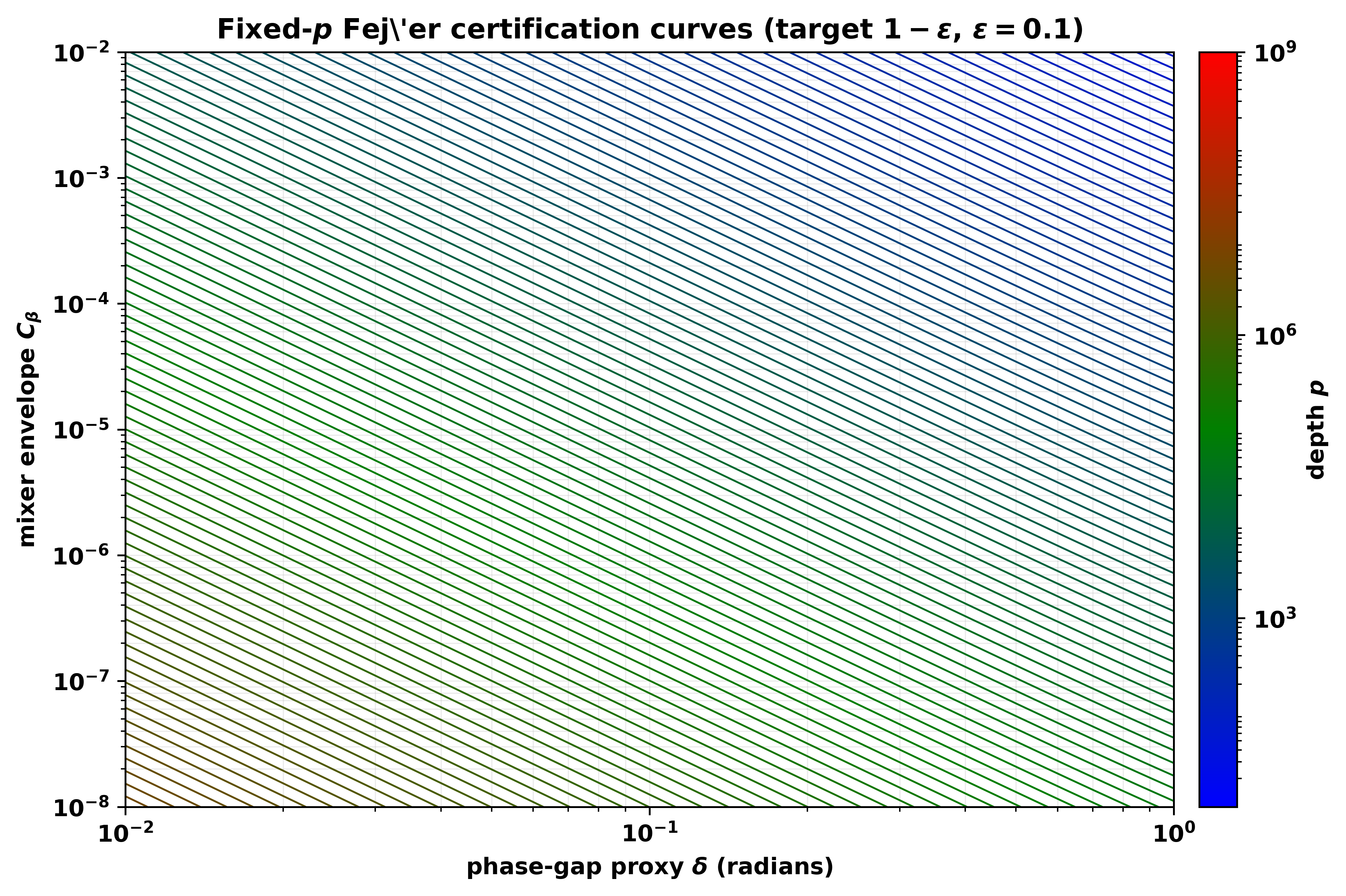}
  \caption{\textbf{Fixed-$p$ certification curves for the sufficient Fej\'er peaking bound}
  (target $1-\varepsilon$ with $\varepsilon=\eps$).
  For each depth $p$, the bound certifies $\Pr[\Omega^\star]\ge 1-\varepsilon$ whenever
  $C_\beta \ge C_{\min}(\delta;\varepsilon,p)$, i.e.\ in the region \emph{above} the corresponding curve.
  The vertical colorbar encodes $p$ from $10^1$ (blue) to $10^{10}$ (red).
  \textbf{Monotonicity:} increasing either the envelope mass $C_\beta$ or the phase-gap proxy $\delta$ \emph{reduces} the certified depth.
  \textbf{Conservatism:} because $\delta$ is a \emph{wrapped} phase-separation proxy, very small values
  (e.g.\ $\delta\sim 10^{-2}$) already correspond to near phase-collisions. In this regime the bound is pessimistic and can predict large orders. Optimistic (too-large) estimates of $C_{\beta}$ or $\delta$ can underpredict the required depth and risk missing the optimum, while conservative (smaller) estimates only inflate the certified depth.  Even with tiny $C_\beta$, the extreme depth region is not approached as long as $\delta$ remains under control.}
  \label{fig:cert-curves-colorbar}
\end{figure}

\subsection{Optimality and Success guarantee}




\begin{theorem}[Dimension-free success bound from Fej\'er filtering]
\label{thm:dimension-free-success}
Assume that the filtered reference distribution in
Eq.~\eqref{eq:factorize} is given pointwise by
\begin{equation}
\label{eq:reference-law-normalized}
\Pr_p^{\mathrm{ref}}[z]
=
\frac{
W_p(z;\boldsymbol\beta)
F_p(\theta(z)-\theta^\star)
}{
\displaystyle
\sum_{y\in\mathcal H_{\mathrm{OH}}}
W_p(y;\boldsymbol\beta)
F_p(\theta(y)-\theta^\star)
},
\qquad
z\in\mathcal H_{\mathrm{OH}}.
\end{equation}
Then the probability of sampling an optimum satisfies
\begin{equation}
\label{eq:dimension-free-success}
q_0
:=
\Pr_p^{\mathrm{ref}}[\Omega^\star]
=
\sum_{z\in\Omega^\star}
\Pr_p^{\mathrm{ref}}[z]
\ge
\frac{(p+1)C_\beta}
{(p+1)C_\beta+M_p(\delta)(1-C_\beta)},
\end{equation}
where
\[
C_\beta
=
\sum_{z\in\Omega^\star}
W_p(z;\boldsymbol\beta)
\]
and $M_p(\delta)$ is defined in
Eq.~\eqref{eq:Mp-wrapped}.
\end{theorem}

\begin{proof}
Let
\[
D:=\sum_{y\in\OH}W_p(y;\boldsymbol\beta)\,F_p(\theta(y)-\theta^\star).
\]
Summing the factorization over $z\in\Omega^\star$ and using $F_p(0)=p+1$ gives
\[
q_0:=\Pr_p[\Omega^\star]
=\frac{\sum_{z\in\Omega^\star}W_p(z;\boldsymbol\beta)\,F_p(0)}{D}
=\frac{(p+1)\,C_\beta}{D}.
\]
By the phase-gap assumption and Lem. \ref{lem:fejer-offpeak}, $F_p(\theta(y)-\theta^\star)\le M_p(\delta)$ for all
$y\notin\Omega^\star$, hence
\[
D
=\sum_{z\in\Omega^\star}W_p(z)\,F_p(0)+\sum_{y\notin\Omega^\star}W_p(y)\,F_p(\theta(y)-\theta^\star)
\le (p+1)C_\beta + M_p(\delta)\sum_{y\notin\Omega^\star}W_p(y).
\]
Since $\sum_{y\in\OH}W_p(y)=1$, we have $\sum_{y\notin\Omega^\star}W_p(y)=1-C_\beta$, so
\[
D\le (p+1)C_\beta + M_p(\delta)(1-C_\beta).
\]
All terms are positive, so $D\le B$ implies $1/D\ge 1/B$, yielding the claim.
\end{proof}

It is important to remark that although our analyses assume a wrapped phase separation, when multiple near-optimal levels cluster in phase, one may replace a strict gap by a ``soft'' phase-separation condition and obtain analogous bounds with $M_p(\delta)$ replaced by the corresponding off-peak mass. This is again supported via  a Riemann--Lebesgue averaging mechanism discussed in  Sec.~\ref{subsec:rl-averaging}.

\subsubsection{Finite depth bound}
A central point of the Fej\'er analysis is that the filter order required to certify \(q_0\ge 1-\epsilon\) is finite whenever there is a fixed phase gap. Using the wrapped off-peak bound from
Lemma~\ref{lem:fejer-offpeak},
\begin{equation}
M_p(\delta)
=
\max_{\substack{\vartheta\in[0,2\pi)\\
\operatorname{dist}_{\mathbb T}(\vartheta,0)\ge\delta}}
F_p(\vartheta)
\le
\frac{1}{(p+1)\sin^2(\delta/2)}.
\end{equation}

Plugging this into the lower bound on the probability of success
\[
q_0
\;\ge\;
\frac{(p+1)C_\beta}{(p+1)C_\beta + M_p(\delta)\,(1-C_\beta)},
\]
we can obtain a sufficient depth for the filter to be peaked on $\Omega^\star$.

Indeed, requiring $q_0\ge 1-\varepsilon$ is equivalent to
$M_p(\delta)(1-C_\beta)\le \frac{\varepsilon}{1-\varepsilon}(p+1)C_\beta$, hence it suffices that
\[
(p+1)^2
\;\ge\;
\frac{1-\varepsilon}{\varepsilon}\cdot \frac{1-C_\beta}{C_\beta}\cdot \csc^2(\delta/2).
\]
Therefore the Fej\'er order needed to peak at the optimum is finite and can be chosen as
\begin{equation}
\label{eq:depth}
p_{\mathrm{peak}}(\varepsilon)
\;:=\;
\Biggl\lceil
\sqrt{\frac{1-\varepsilon}{\varepsilon}\cdot \frac{1-C_\beta}{C_\beta}}\;\csc(\delta/2)
\Biggr\rceil \;-\; 1,
\end{equation}
which depends only on the phase gap $\delta$ (and the envelope constant $C_\beta$), and not on the ambient Hilbert-space dimension.
In particular, for a constant target $\varepsilon$ and constant gap $\delta=\Omega(1)$, we obtain $p_{\mathrm{peak}}=O(1)$. See Figure \ref{fig:cert-curves-colorbar}.

\subsubsection{Finite shot regimes}
Since the probability of sampling an optimum \(z^\star\) satisfies
\begin{equation}
\label{eq:Prstar-recap}
q_0
\;\ge\;
\frac{(p+1) C_\beta}{(p+1) C_\beta + M_p(\delta)\,(1-C_\beta)},
\end{equation}
introduce the single control parameter
\begin{equation}
\label{eq:def-x}
A\;:=\;(p{+}1)^2 \sin^2(\delta/2),
\qquad
x\;:=\;A\,C_\beta.
\end{equation}
Multiplying numerator and denominator of \eqref{eq:Prstar-recap} by \((p{+}1)\sin^2(\delta/2)\)
yields the \emph{ratio form}
\begin{equation}
\label{eq:q0-ratio}
q_0\;
\;\ge\;
\frac{x}{(1-C_\beta)+x}
\;\ge\;
\frac{x}{1+x},
\end{equation}
where the last inequality uses \(1-C_\beta\le 1\). 

Since
\[
(1-q_0)^S\le e^{-q_0S}
\]
and $q_0\ge x/(1+x)$, the sufficient certified shot budget is
\begin{equation}
\label{eq:certified-shot-budget}
S_{\mathrm{cert}}
:=
\left\lceil
\left(1+\frac{1}{x}\right)
\ln\frac{1}{\epsilon}
\right\rceil.
\end{equation}
and the standard shot bound
is entirely controlled by \(x=(p{+}1)^2\sin^2(\delta/2)C_\beta\).

Thus, three regimes of the shot budgets can be easily identified.

\paragraph{(R1) Small–product regime \(x\ll 1\) (weak gap and/or tiny envelope).}
Since \( x/(1+x)\approx x + O(x^2)\),
\begin{equation}
\label{eq:regime-small}
q_0
\ge
\frac{x}{1+x}
=
x+O(x^2),
\qquad
S_{\mathrm{cert}}
=
O\!\left(
\frac{1}{x}\ln\frac{1}{\epsilon}
\right)
=
O\!\left(
\frac{1}{
(p+1)^2\sin^2(\delta/2)C_\beta
}
\ln\frac{1}{\epsilon}
\right).
\end{equation}
For $\delta\ll1$, $\sin(\delta/2)\sim\delta/2$, and hence
\[
S_{\mathrm{cert}}
=
\Theta\!\left(
\frac{1}{(p+1)^2\delta^2 C_\beta}
\ln\frac{1}{\epsilon}
\right).
\]
Thus the certified shot budget is quadratically sensitive to the inverse phase gap and linearly sensitive to the inverse mixer-envelope mass. Most importantly, neither a small phase gap nor a small
mixer-envelope mass precludes a finite-shot guarantee. The filter
order, phase separation, envelope mass, and shot budget enter through
the compensatory relation
\begin{equation}
\label{eq:finite-synth}
S_{\mathrm{cert}}\,
(p+1)^2\sin^2(\delta/2)\,C_\beta
=
O\!\left(\ln\frac{1}{\epsilon}\right).
\end{equation}
Thus a small $\delta$ or $C_\beta$ can be offset by increasing the filter order and the number of shots.

\paragraph{(R2) Threshold regime \(x\approx 1\) (``knife--edge'').}
Let \(x\in[1-\eta,\,1+\eta]\) for \(0<\eta<1\). Then from \eqref{eq:q0-ratio},
\begin{equation}
\label{eq:regime-threshold}
q_0\ \ge\ \frac{1-\eta}{2-\eta}
\quad\Longrightarrow\quad
S_{cert}\ \le\ \frac{2-\eta}{1-\eta}\,\ln\frac{1}{\epsilon}
\ =\ \bigl(2+O(\eta)\bigr)\,\ln(1/\epsilon).
\end{equation}
Thus, even at the threshold \(x\simeq 1\), the shot complexity remains \emph{bounded and
dimension–free}, scaling only with \(\ln(1/\epsilon)\).

\paragraph{(R3) Large–product regime \(x\gg 1\) (healthy gap and good envelope).}
Now \(q_0\ge x/(1+x)=1-1/(1+x)=1-O(1/x)\), so
\begin{equation}
\label{eq:regime-large}
q_0\ \approx\ 1,
\qquad
S_{cert}\ \approx\ \ln\frac{1}{\epsilon}.
\end{equation}
Once \(x\) exceeds a modest constant (say \(x\gtrsim 10\)), the shots are essentially optimal.

\subsection{Fej\'er Factored Feasibility Guarantee}
\label{sec:fejfeas}
Recall that Theorem~\ref{thm:finite-depth-feasible}  asserts that feasibility can be made constant at finite depth by a CE--QAOA parameter choice exploiting the kernel structure. A quantitative finite-depth lower bound follows from a specialization of the Fej\'er factorization law
(Eq.~\ref{eq:factorize}) when we use \emph{penalty-only phases} $H_C\equiv H_{\mathrm{pen}}$ (App.~\ref{sec:fejer-feasibility}).  In particular, Cor.~\ref{thm:fejer-feasibility} yields
\[
\pi_{\mathsf F}^{\mathrm{ref}}(p;\boldsymbol\beta,\gamma)
\ \ge\
\frac{x_{\mathsf F}}{1+x_{\mathsf F}},
\qquad
x_{\mathsf F}:=(p+1)^2\sin^2(\delta_{\mathsf F}/2)\,C^{\mathsf F}_{p,\boldsymbol\beta},
\]
where $C^{\mathsf F}_{p,\boldsymbol\beta}=\sum_{z\in L_0}W_p(z;\boldsymbol\beta)$ is the mixer-envelope
mass on the feasible set and $\delta_{\mathsf F}$ is the penalty-phase separation
defined in~\eqref{eq:delta-feasible}. Of particular interest are the shallow-depth cases
\[
p=1:\quad \pi_{\mathsf F}^{\mathrm{ref}}\ \ge\ \frac{4\sin^2(\delta_{\mathsf F}/2)\,C^{\mathsf F}_{1,\boldsymbol\beta}}{1+4\sin^2(\delta_{\mathsf F}/2)\,C^{\mathsf F}_{1,\boldsymbol\beta}},
\qquad
p=2:\quad \pi_{\mathsf F}^{\mathrm{ref}}\ \ge\ \frac{9\sin^2(\delta_{\mathsf F}/2)\,C^{\mathsf F}_{2,\boldsymbol\beta}}{1+9\sin^2(\delta_{\mathsf F}/2)\,C^{\mathsf F}_{2,\boldsymbol\beta}}.
\]

\begin{figure}
\centering
\begin{tikzpicture}
\pgfplotsset{compat=1.17}

\pgfplotsset{
  colormap={phase}{
    rgb(0cm)=(0.90,1.00,0.90)  
    rgb(0.5cm)=(0.20,0.40,1.00)
    rgb(1cm)=(0.00,0.40,0.00)  
  }
}

\begin{axis}[
    view={0}{90},                 
    xlabel={$C_\beta$},
    ylabel={$\delta$},
    xmin=0, xmax=1,
    ymin=0, ymax=1.57,            
    domain=0:1,
    domain y=0:1.57,
    samples=80,
    axis on top,
    colorbar,
    colorbar style={
      title={$q_0$},
    },
    title={Phase diagram of single-shot success $q_0$},
]

\def\pplus{2}

\addplot3 [
    surf,
    shader=interp,
    colormap name=phase,
    point meta={
      ( (\pplus*\pplus) * (sin(deg(y/2))^2) * x ) /
      ( 1 + (\pplus*\pplus) * (sin(deg(y/2))^2) * x )
    },
] {0};

\node[anchor=west] at (rel axis cs:0.05,0.15) {(R1) small $x$};
\node[anchor=west] at (rel axis cs:0.40,0.45) {(R2) $x \approx 1$};
\node[anchor=west] at (rel axis cs:0.70,0.85) {(R3) large $x$};

\end{axis}
\end{tikzpicture}
\caption{Phase diagram for the Fej\'er-based lower bound.
Deeper green indicates higher single-shot success $q_0$; pale green indicates weaker performance.}
\label{fig:phase-diagram}
\end{figure}
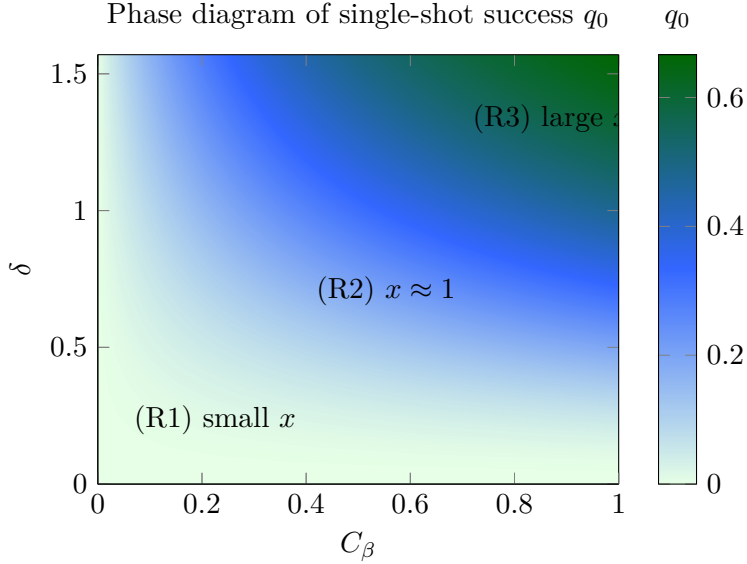

\section{Robustness Beyond Lattice Normalization}
\label{subsec:rl-averaging}

\subsection{Riemann--Lebesgue averaging beyond exact lattice normalization}

Our baseline Fej\'er analysis is most transparent when the rescaled cost spectrum
lies on an integer lattice and one can select a harmonic schedule that yields a
clean wrapped phase gap. In practice, however, one often encounters objectives
whose natural normalization is only \emph{approximately} compatible with a
discrete lattice, for example due to floating-point weights, empirical penalty calibration, or aggregation of heterogeneous cost components. Empirically, the Fej\'er mechanism remains useful well beyond the exact-lattice idealization.

Instead of insisting on exact commensurability, we
\emph{average} the filtered protocol over a small one-parameter family of cost
scales (or equivalently, small perturbations of the base cost angle). The
resulting effective kernel remains \emph{nonnegative} and retains a
\emph{dimension-free} success bound. The suppression mechanism is the
Riemann--Lebesgue lemma, which states that oscillatory Fourier integrals vanish
at nonzero frequencies under mild integrability conditions
\cite{SteinShakarchi2003Fourier,Katznelson2004}.

\paragraph{Set-up.}
Fix a diagonal cost Hamiltonian $H_C$ with eigenvalues $E(z)$ on the computational
basis $\{\ket{z}\}_{z\in\mathcal X}$. Let $\Omega^\star$ denote the set of optimal
strings with energy $E^\star=\min_z E(z)$.
Let $W_p(\cdot;\boldsymbol\beta)$ be the mixer envelope from
Eq.~\eqref{eq:Wp-def} (or any envelope distribution produced by a classicalized
mixer evolution). As in Section~\ref{sec:fejer}, define the \emph{centered} normalized Dirichlet polynomial
\[
D_p^\star(\gamma;H_C)
\;:=\;
\frac{1}{\sqrt{p+1}}
\sum_{r=0}^{p} e^{-i r \gamma (H_C-E^\star I)}.
\]
Equivalently, one may write
\[
D_p^\star(\gamma;H_C)
\;=\;
\frac{1}{\sqrt{p+1}}
\sum_{r=0}^{p} e^{-ir(\gamma H_C-\theta^\star)},
\qquad
\theta^\star:=\gamma E^\star \ (\mathrm{mod}\ 2\pi).
\]
For a fixed $\gamma$, the diagonal action of
$D_p^\star(\gamma;H_C)^\dagger D_p^\star(\gamma;H_C)$
produces the Fej\'er kernel evaluated at the centered phase offset
\[
\theta(z)-\theta^\star
=
\gamma(E(z)-E^\star)\ (\mathrm{mod}\ 2\pi).
\]

\paragraph{Averaging (dithering) in the cost angle.}
Let $w$ be a probability measure on $\mathbb R$ with $w\in L^1(\mathbb R)$ and
$\int_\mathbb R w(u)\,du=1$. We define a \emph{dithered} (randomized) base cost
angle by drawing $u\sim w$ and using $\gamma'=\gamma+u$ inside the centered Dirichlet filter.
Operationally, this corresponds to running the same mixer-prepared envelope
$\rho_{\mathrm{env}}$ and applying $D_p^\star(\gamma+u;H_C)$, then averaging statistics
over $u$ (or, equivalently, averaging over a small coarse grid in $\gamma$).

\begin{definition}[Averaged Fej\'er kernel in energy differences]
\label{def:avg-fejer}
For $\Delta E\in\mathbb R$, define the \emph{averaged} (dithered) Fej\'er weight
\begin{equation}
\label{eq:avg-fejer-def}
\overline F_{p,w}^{(\gamma)}(\Delta E)
\;:=\;
\int_{\mathbb R} w(u)\,
F_p\!\bigl((\gamma+u)\Delta E\bigr)\,du,
\end{equation}
where $F_p$ is the Fej\'er kernel from Eq.~\eqref{eq:fejer-def}.
\end{definition}

\noindent
By positivity of $F_p$ and $w$, $\overline F_{p,w}^{(\gamma)}(\Delta E)\ge 0$ for all $\Delta E$.
Moreover, $\overline F_{p,w}^{(\gamma)}(0)=F_p(0)=p+1$.

\paragraph{Fourier form and the Riemann--Lebesgue suppression factor.}
Write the Fourier transform of $w$ as
\[
\widehat w(\xi)
\;:=\;
\int_{\mathbb R} w(u)\,e^{i\xi u}\,du.
\]
Since $w\in L^1(\mathbb R)$, the Riemann--Lebesgue lemma implies
$\widehat w(\xi)\to 0$ as $|\xi|\to\infty$ \cite{SteinShakarchi2003Fourier,Katznelson2004}.
Expanding the squared Dirichlet sum and averaging term-by-term yields an explicit
representation of $\overline F_{p,w}^{(\gamma)}$ in terms of $\widehat w$.

\begin{lemma}[Averaged Dirichlet square and off-peak control]
\label{lem:avg-fejer-fourier}
For every $\Delta E\in\mathbb R$,
\begin{align}
\label{eq:avg-fejer-fourier}
\overline F_{p,w}^{(\gamma)}(\Delta E)
&=
\frac{1}{p+1}\sum_{r,s=0}^{p}
e^{-i(r-s)\gamma\Delta E}\,
\widehat w\bigl((r-s)\Delta E\bigr),
\\
\label{eq:avg-fejer-bound-general}
\overline F_{p,w}^{(\gamma)}(\Delta E)
&\le
1
+
2\sum_{k=1}^{p}\Bigl(1-\frac{k}{p+1}\Bigr)\,
\bigl|\widehat w(k\Delta E)\bigr|.
\end{align}
In particular, for any gap $g>0$,
\begin{equation}
\label{eq:Mp-avg-def}
\overline M_{p,w}(g)
\;:=\;
\sup_{|\Delta E|\ge g}\ \overline F_{p,w}^{(\gamma)}(\Delta E)
\;\le\;
1
+
2\sum_{k=1}^{p}\Bigl(1-\frac{k}{p+1}\Bigr)\,
\sup_{|\Delta E|\ge g}\bigl|\widehat w(k\Delta E)\bigr|.
\end{equation}
\end{lemma}

\begin{proof}
Start from
\[
F_p(x)
=\frac{1}{p+1}\Bigl|\sum_{r=0}^{p}e^{-irx}\Bigr|^2
=\frac{1}{p+1}\sum_{r,s=0}^{p} e^{-i(r-s)x}.
\]
Substitute $x=(\gamma+u)\Delta E$ and average over $u$ with weight $w(u)$.
This yields \eqref{eq:avg-fejer-fourier} with $\widehat w((r-s)\Delta E)$.
For \eqref{eq:avg-fejer-bound-general}, group terms by $k=r-s$.
There are $(p+1-k)$ pairs $(r,s)$ with $r-s=k$ and the same number with $r-s=-k$.
Use the triangle inequality and $|\widehat w(-\xi)|=|\widehat w(\xi)|$.
\end{proof}

\paragraph{Uniform dithering.}
A convenient choice is a uniform window $w(u)=\frac{1}{2\Gamma}\mathbf 1_{[-\Gamma,\Gamma]}(u)$. Then
\[
\widehat w(\xi)
=
\frac{\sin(\Gamma\xi)}{\Gamma\xi},
\qquad
|\widehat w(\xi)|\le \min\Bigl\{1,\frac{1}{\Gamma|\xi|}\Bigr\}.
\]
Plugging this into Lemma~\ref{lem:avg-fejer-fourier} yields the explicit off-peak estimate
\begin{equation}
\label{eq:avg-fejer-uniform}
\overline M_{p,w}(g)
\;\le\;
1
+
\frac{2}{\Gamma g}\sum_{k=1}^{p}\Bigl(1-\frac{k}{p+1}\Bigr)\frac{1}{k}
\;\le\;
1
+
\frac{2\log(p+1)}{\Gamma g}.
\end{equation}
Thus, even without any exact lattice structure, a modest averaging window
$\Gamma=\Theta(\log(p)/g)$ makes $\overline M_{p,w}(g)=1+O(1)$.

\paragraph{A dimension-free success bound with nonlattice costs.}
We now state the analogue of Theorem~\ref{thm:dimension-free-success}, where the
wrapped phase gap is replaced by an ordinary (non-wrapped) \emph{energy gap} and
off-peak suppression is provided by Riemann--Lebesgue averaging.

\begin{assumption}[Energy separation]
\label{assump:energy-gap}
There exists $g>0$ such that for all $y\notin \Omega^\star$,
\[
|E(y)-E^\star|\ \ge\ g.
\]
\end{assumption}

\begin{theorem}[Dimension-free bound under Riemann--Lebesgue averaging]
\label{thm:rl-dimension-free}
Assume the filtered reference law factorizes as
\[
\Pr[z]
=
\frac{W_p(z;\boldsymbol\beta)\,F_p\bigl((\gamma+u)(E(z)-E^\star)\bigr)}
{\sum_{y\in\mathcal X}W_p(y;\boldsymbol\beta)\,F_p\bigl((\gamma+u)(E(y)-E^\star)\bigr)},
\]
where $u$ is drawn independently from a density $w\in L^1(\mathbb R)$.
Let $\overline{\Pr}$ denote the resulting measurement law averaged over $u$.
Define
\[
C_\beta \;:=\; \sum_{z\in\Omega^\star} W_p(z;\boldsymbol\beta),
\qquad
\overline M_{p,w}(g) \;:=\; \sup_{|\Delta E|\ge g}\ \overline F_{p,w}^{(\gamma)}(\Delta E),
\]
where $\overline F_{p,w}^{(\gamma)}$ is from Def.~\ref{def:avg-fejer}.
Under Assumption~\ref{assump:energy-gap},
\begin{equation}
\label{eq:rl-success}
\overline{\Pr}[\Omega^\star]
\;\ge\;
\frac{(p+1)\,C_\beta}{(p+1)\,C_\beta+\overline M_{p,w}(g)\,(1-C_\beta)}.
\end{equation}
Consequently, with
\[
x_{\mathrm{RL}}
\;:=\;
\frac{(p+1)\,C_\beta}{\overline M_{p,w}(g)},
\]
we have the ratio form
\[
\overline{\Pr}[\Omega^\star]
\;\ge\;
\frac{x_{\mathrm{RL}}}{(1-C_\beta)+x_{\mathrm{RL}}}
\;\ge\;
\frac{x_{\mathrm{RL}}}{1+x_{\mathrm{RL}}},
\]
and the standard shot bound

\[
S_{\mathrm{cert}}^{\mathrm{RL}}
:=
\left\lceil
\left(1+\frac{1}{x_{\mathrm{RL}}}\right)
\ln\frac{1}{\epsilon}
\right\rceil.
\]

remains dimension-free.
\end{theorem}


\begin{proof}
For each fixed $u$, write
\[
\Pr_u[z]
=
\frac{W_p(z;\boldsymbol\beta)\,
F_p\bigl((\gamma+u)(E(z)-E^\star)\bigr)}
{\sum_{y\in\mathcal X}W_p(y;\boldsymbol\beta)\,
F_p\bigl((\gamma+u)(E(y)-E^\star)\bigr)}.
\]
Then the averaged law is
\[
\overline{\Pr}[A]
=
\int_{\mathbb R} w(u)\,\Pr_u[A]\,du
=
\mathbb E_u[\Pr_u[A]]
\]
for every event $A\subseteq\mathcal X$.

Set
\[
A:=(p+1)\,C_\beta,
\qquad
B(u):=
\sum_{y\notin\Omega^\star}
W_p(y;\boldsymbol\beta)\,
F_p\bigl((\gamma+u)(E(y)-E^\star)\bigr).
\]
Since $F_p(0)=p+1$ on $\Omega^\star$, for each fixed $u$ we have
\[
\Pr_u[\Omega^\star]
=
\frac{A}{A+B(u)}.
\]
Hence
\[
\overline{\Pr}[\Omega^\star]
=
\mathbb E_u\!\left[\frac{A}{A+B(u)}\right].
\]

Now the function
\[
f(x):=\frac{A}{A+x},\qquad x\ge 0,
\]
is convex because
\[
f''(x)=\frac{2A}{(A+x)^3}>0.
\]
Therefore Jensen's inequality gives
\[
\overline{\Pr}[\Omega^\star]
=
\mathbb E_u[f(B(u))]
\;\ge\;
f(\mathbb E_u[B(u)])
=
\frac{A}{A+\mathbb E_u[B(u)]}.
\]

It remains to bound $\mathbb E_u[B(u)]$. By Tonelli/Fubini and the definition of
$\overline F_{p,w}^{(\gamma)}$,
\begin{align*}
\mathbb E_u[B(u)]
&=
\sum_{y\notin\Omega^\star}
W_p(y;\boldsymbol\beta)\,
\mathbb E_u\!\left[
F_p\bigl((\gamma+u)(E(y)-E^\star)\bigr)
\right]
\\
&=
\sum_{y\notin\Omega^\star}
W_p(y;\boldsymbol\beta)\,
\overline F_{p,w}^{(\gamma)}\bigl(E(y)-E^\star\bigr).
\end{align*}
Under Assumption~\ref{assump:energy-gap}, every $y\notin\Omega^\star$ satisfies
\[
|E(y)-E^\star|\ge g,
\]
so by the definition of $\overline M_{p,w}(g)$,
\[
\overline F_{p,w}^{(\gamma)}\bigl(E(y)-E^\star\bigr)
\le
\overline M_{p,w}(g).
\]
Therefore
\[
\mathbb E_u[B(u)]
\le
\overline M_{p,w}(g)\,
\sum_{y\notin\Omega^\star}W_p(y;\boldsymbol\beta)
=
\overline M_{p,w}(g)\,(1-C_\beta).
\]
Substituting into the Jensen bound yields
\[
\overline{\Pr}[\Omega^\star]
\;\ge\;
\frac{(p+1)\,C_\beta}
{(p+1)\,C_\beta+\overline M_{p,w}(g)\,(1-C_\beta)},
\]
which is \eqref{eq:rl-success}. The ratio form then follows as in the proof of Theorem~\ref{thm:dimension-free-success}.
\end{proof}

The RL-averaged bound trades the wrapped phase-gap proxy $\delta$ for an ordinary
energy gap $g$ and a controllable averaging window that suppresses oscillatory
cross-terms through $\widehat w(k\Delta E)$.
This is exactly the regime where exact lattice normalization is awkward but
small scale randomization of $\gamma$ is cheap.
It also matches the empirical observation that small deviations from perfect
commensurability often do not destroy performance, because coarse averaging washes out resonant aliasing effects. The averaging can be implemented at the data-collection level. This realizes $\overline{\Pr}$ directly and makes the bound applicable even when
$H_C$ is not exactly lattice-normalized.

\subsection{Jensen's inequality and the averaged success bound}
\label{subsec:jensen-pedagogical}

A small but important point in the Riemann--Lebesgue averaged analysis is that the randomized success probability is an \emph{average of normalized fractions}, not the normalization of an averaged numerator and denominator. This is where Jensen's inequality enters. For each realization of the dither variable $u$, the success probability of the
optimal set has the form
\begin{equation}
\label{eq:jensen-basic-success}
\Pr_u[\Omega^\star]
=
\frac{A}{A+B(u)},
\end{equation}
where
\begin{equation}
\label{eq:jensen-A-B}
A:=(p+1)C_\beta,
\qquad
B(u):=
\sum_{y\notin\Omega^\star}
W_p(y;\boldsymbol\beta)\,
F_p\!\bigl((\gamma+u)(E(y)-E^\star)\bigr).
\end{equation}
Here $A$ is the total optimal contribution, which is independent of $u$, while
$B(u)$ is the nonoptimal contribution, which fluctuates with the random angle
perturbation. Averaging over $u$ gives
\begin{equation}
\label{eq:jensen-avg-success}
\overline{\Pr}[\Omega^\star]
=
\mathbb E_u\!\left[\frac{A}{A+B(u)}\right].
\end{equation}
At this point one cannot simply move the expectation inside the ratio and write
\[
\mathbb E_u\!\left[\frac{A}{A+B(u)}\right]
\stackrel{\text{in general}}{\neq}
\frac{A}{A+\mathbb E_u[B(u)]}.
\]
The average of a nonlinear function is generally not the function of the average.

\paragraph{The effect of nonlinearity helps.}
The map
\begin{equation}
\label{eq:jensen-f-def}
f(t):=\frac{A}{A+t},
\qquad t\ge 0,
\end{equation}
is decreasing, but more importantly it is \emph{convex}:
\begin{equation}
\label{eq:jensen-second-deriv}
f''(t)=\frac{2A}{(A+t)^3}>0.
\end{equation}
Convexity means that the graph of $f$ lies \emph{below} its secant lines, and
Jensen's inequality therefore gives
\begin{equation}
\label{eq:jensen-main-step}
\mathbb E_u[f(B(u))]
\;\ge\;
f(\mathbb E_u[B(u)]).
\end{equation}
Substituting the present $f$ yields
\begin{equation}
\label{eq:jensen-applied}
\overline{\Pr}[\Omega^\star]
=
\mathbb E_u\!\left[\frac{A}{A+B(u)}\right]
\;\ge\;
\frac{A}{A+\mathbb E_u[B(u)]}.
\end{equation}

Although expectation does not pass through the ratio exactly, convexity of the reciprocal-type map provides the inequality in the
\emph{correct direction} for a lower bound.

\paragraph{Intuition.}
The function $t\mapsto A/(A+t)$ is a reciprocal-type damping factor. Because it is convex, fluctuations of $B(u)$ upward and downward do not cancel linearly: large values of $B(u)$ hurt success, but small values help by a slightly larger amount in the averaged reciprocal scale. Thus replacing the random denominator by its mean gives a conservative lower bound.

This is closely related to the familiar fact that reciprocals favor averaging in
the harmonic direction rather than the arithmetic direction. In the present
setting, the success probability is a normalized positive weight, so the
relevant nonlinear structure is exactly of this reciprocal form. Once \eqref{eq:jensen-applied} is established, the only task is to control the
mean nonoptimal weight $\mathbb E_u[B(u)]$. By Tonelli/Fubini,
\begin{align}
\label{eq:jensen-tonelli-step}
\mathbb E_u[B(u)]
&=
\sum_{y\notin\Omega^\star}
W_p(y;\boldsymbol\beta)\,
\mathbb E_u\!\left[
F_p\!\bigl((\gamma+u)(E(y)-E^\star)\bigr)
\right]
\nonumber\\
&=
\sum_{y\notin\Omega^\star}
W_p(y;\boldsymbol\beta)\,
\overline F_{p,w}^{(\gamma)}\!\bigl(E(y)-E^\star\bigr).
\end{align}
If every nonoptimal configuration satisfies the energy-gap condition
$|E(y)-E^\star|\ge g$, then
\[
\overline F_{p,w}^{(\gamma)}\!\bigl(E(y)-E^\star\bigr)
\le
\overline M_{p,w}(g),
\]
and hence
\begin{equation}
\label{eq:jensen-final-B-bound}
\mathbb E_u[B(u)]
\le
\overline M_{p,w}(g)\,(1-C_\beta).
\end{equation}
Substituting this into \eqref{eq:jensen-applied} gives the averaged dimension-free lower bound.

\subsection{Order Reduction from Lipschitz/Main--Lobe Control}
\label{sec:depth-reduction}

\paragraph{Normalization and norm scales.}
To avoid hidden dimension- or instance-dependent constants, we fix
explicit operator-norm normalizations for all generators. Let
\(\|\cdot\|\) denote the operator norm. We introduce known scales
\(\Lambda_C,\Lambda_{\mathrm{pen}},\Lambda_M>0\) such that
\[
\|H_C\|\le \Lambda_C,
\qquad
\|H_{\mathrm{pen}}\|\le \Lambda_{\mathrm{pen}},
\qquad
\|H_M\|\le \Lambda_M,
\]
and work with the normalized Hamiltonians
\[
\widetilde H_C
:=
\frac{H_C}{\Lambda_C},
\qquad
\widetilde H_{\mathrm{pen}}
:=
\frac{H_{\mathrm{pen}}}{\Lambda_{\mathrm{pen}}},
\qquad
\widetilde H_M
:=
\frac{H_M}{\Lambda_M},
\]
so that
\[
\|\widetilde H_C\|,
\|\widetilde H_{\mathrm{pen}}\|,
\|\widetilde H_M\|
\le 1.
\]
All angles are interpreted with respect to these normalized generators:
\[
e^{-i\gamma H_C}
=
e^{-i\widetilde\gamma\,\widetilde H_C},
\qquad
\widetilde\gamma
:=
\gamma\Lambda_C,
\]
and similarly for the mixer and penalty angles. The Lipschitz bounds, Fej\'er off-peak estimates, and shot-complexity
expressions are thus not affected by any scaling in \(\|H\|\).

\begin{lemma}[Lipschitz and main--lobe control at reduced order]
\label{lem:lip-reduced}
Let \(p'=p-k\) with \(k\in\{1,\dots,p-1\}\). Then:
\begin{enumerate}[leftmargin=1.5em]

\item \textbf{Mixer-envelope mass on the optimum set.}
Define
\[
C_{\beta}^{(p')}
:=
\sum_{z\in\Omega^\star}
W_{p'}(z;\boldsymbol\beta).
\]
For any \(\boldsymbol\beta\) and perturbation
\(\Delta\boldsymbol\beta\),
\[
\left|
C_{\beta+\Delta\beta}^{(p')}
-
C_{\beta}^{(p')}
\right|
\le
L_W^{(p')}
\|\Delta\boldsymbol\beta\|_\infty,
\qquad
L_W^{(p')}
=
O\!\bigl(p'\|H_M\|\bigr).
\]

\item \textbf{Fej\'er main lobe (constant-fraction capture).}
The Fej\'er kernel
\[
F_{p'}(\Delta)
=
\frac{1}{p'+1}
\frac{
\sin^2\!\bigl((p'+1)\Delta/2\bigr)
}{
\sin^2(\Delta/2)
}
\]
has its nearest nonzero zeros at
\[
|\Delta|
=
\frac{2\pi}{p'+1}
\]
and hence has main-lobe width \(\Theta(1/p')\). For every fixed
\(c\in(0,\pi)\),
\[
|\Delta\theta|
\le
\frac{c}{p'+1}
\quad\Longrightarrow\quad
F_{p'}(\Delta\theta)
\ge
\kappa_c(p'+1),
\]
where
\[
\kappa_c
:=
\left(
\frac{\sin(c/2)}{c/2}
\right)^2
\in(0,1).
\]

Equivalently, let
\[
\Delta\theta
=
\Delta\gamma\,
\bigl(E(z)-E^\star\bigr)
\]
and choose
\[
R
\ge
\max_z
\bigl|E(z)-E^\star\bigr|.
\]
Then
\[
|\Delta\gamma|
\le
\frac{c}{(p'+1)R}
\quad\Longrightarrow\quad
F_{p'}(\Delta\theta)
\ge
\kappa_c(p'+1).
\]

\end{enumerate}
\end{lemma}

\begin{proof}[Proof sketch]
\emph{(1) Mixer-envelope mass.}
Write
\[
U_M^{(p')}(\boldsymbol\beta)
=
\prod_{\ell=1}^{p'}
e^{-i\beta_\ell H_M}
\]
and introduce the projector onto the full optimum subspace,
\[
\Pi_{\Omega^\star}
:=
\sum_{z\in\Omega^\star}
|z\rangle\langle z|.
\]
Then
\[
C_{\beta}^{(p')}
=
\langle s|
U_M^{(p')}(\boldsymbol\beta)^\dagger
\Pi_{\Omega^\star}
U_M^{(p')}(\boldsymbol\beta)
|s\rangle.
\]
By the Wilcox--Duhamel formula,
\[
\left\|
\frac{\partial}{\partial\beta_\ell}
e^{-i\beta_\ell H_M}
\right\|
\le
\|H_M\|.
\]
A product rule together with unitarity gives
\[
\left\|
\frac{\partial}{\partial\beta_\ell}
U_M^{(p')}(\boldsymbol\beta)
\right\|
\le
\|H_M\|.
\]
Since \(\Pi_{\Omega^\star}\) is an orthogonal projector,
\[
\|\Pi_{\Omega^\star}\|=1,
\]
independently of the degeneracy of the optimum. Therefore,
\[
\left|
C_{\beta+\Delta\beta}^{(p')}
-
C_{\beta}^{(p')}
\right|
\le
2
\sum_{\ell=1}^{p'}
\left\|
\frac{\partial}{\partial\beta_\ell}
U_M^{(p')}(\widetilde{\boldsymbol\beta})
\right\|
|\Delta\beta_\ell|
\le
2p'\|H_M\|
\|\Delta\boldsymbol\beta\|_\infty,
\]
for some \(\widetilde{\boldsymbol\beta}\) on the segment joining
\(\boldsymbol\beta\) and
\(\boldsymbol\beta+\Delta\boldsymbol\beta\). Hence
\[
L_W^{(p')}
=
O\!\bigl(p'\|H_M\|\bigr).
\]

\smallskip
\emph{(2) Fej\'er main lobe.}
For
\[
|\Delta|
\le
\frac{c}{p'+1},
\]
write \(u=(p'+1)\Delta/2\). Then \(|u|\le c/2<\pi/2\), and
\[
F_{p'}(\Delta)
=
(p'+1)
\left(
\frac{\sin u}{u}
\right)^2
\left(
\frac{\Delta/2}{\sin(\Delta/2)}
\right)^2.
\]
The standard \(\sin x/x\) bounds give
\[
F_{p'}(\Delta)
\ge
(p'+1)
\left(
\frac{\sin(c/2)}{c/2}
\right)^2
=
\kappa_c(p'+1).
\]
Finally, if
\[
\Delta\theta
=
\Delta\gamma\,
\bigl(E(z)-E^\star\bigr)
\]
and
\[
R
\ge
\max_z|E(z)-E^\star|,
\]
then
\[
|\Delta\gamma|
\le
\frac{c}{(p'+1)R}
\]
implies
\[
|\Delta\theta|
\le
\frac{c}{p'+1},
\]
which proves the stated phase-resolution bound.
\end{proof}

For notational clarity, write
\[
C_\beta:=C_\beta^{(p)},
\qquad
C'_\beta:=C_\beta^{(p')}.
\]
Two practical policies can maintain
\[
x_{p'}
=
(p'+1)^2
\sin^2(\delta'/2)
C'_\beta
=
\Omega(1).
\]
First, enforce a base-angle floor \(\gamma_{\min}>0\) and rescale
\(\gamma\), through a coarse scale grid in \((0,1]\), so that
\(\delta'=\gamma'\Delta_{\mathrm{lat}}\) stays comparable to \(\delta\).
The exact ratio is
\[
\frac{x_{p'}}{x_p}
=
\left(
\frac{p'+1}{p+1}
\right)^2
\frac{\sin^2(\delta'/2)}{\sin^2(\delta/2)}
\frac{C'_\beta}{C_\beta}.
\]
Consequently, if
\[
\frac{\sin^2(\delta'/2)}{\sin^2(\delta/2)}
\ge c_1,
\qquad
\frac{C'_\beta}{C_\beta}
\ge c_2
\]
for constants \(c_1,c_2>0\), then
\[
\frac{x_{p'}}{x_p}
\ge
c_1c_2
\left(
\frac{p'+1}{p+1}
\right)^2.
\]

The second policy exploits the broader lobes available at lower order.
Since \(F_{p'}\) has a wider main lobe than \(F_p\), the required
\(\gamma\)-resolution is less stringent:
\[
\Theta\!\left(\frac{1}{p'R}\right)
\qquad\text{instead of}\qquad
\Theta\!\left(\frac{1}{pR}\right).
\]
The corresponding shot tradeoff is as follows.

\begin{proposition}[Shot complexity under order reduction]
\label{prop:shots-reduced}
Suppose that, at order \(p\), the ratio parameter satisfies
\[
x_p\ge x_0>0.
\]
Assume that the reduced-order schedule satisfies
\[
\sin^2(\delta'/2)
\ge
c_\delta\sin^2(\delta/2),
\qquad
C'_\beta
\ge
c_\beta C_\beta,
\]
for constants \(c_\delta,c_\beta\in(0,1]\). Then, for
\(p'=p-k\),
\[
x_{p'}
\ge
c'\,x_0
\left(
\frac{p'+1}{p+1}
\right)^2,
\qquad
c'
:=
c_\delta c_\beta.
\]
Consequently, the sufficient certified shot budget satisfies
\[
S_{\mathrm{cert}}^{(p')}
\le
\left\lceil
\left[
1+
\frac{1}{c'x_0}
\left(
\frac{p+1}{p'+1}
\right)^2
\right]
\ln\frac{1}{\epsilon}
\right\rceil.
\]
The certified shot overhead remains bounded by a constant factor
whenever
\[
\frac{p+1}{p'+1}
=
O(1)
\]
and the phase-gap and envelope-mass ratios remain bounded below by
positive constants. In particular, this applies to
\(p'\in\{1,2\}\) when the parent order \(p\) is itself fixed.
\end{proposition}

\begin{proof}
The assumptions imply
\[
\begin{aligned}
x_{p'}
&=
(p'+1)^2
\sin^2(\delta'/2)
C'_\beta
\\
&\ge
c_\delta c_\beta
\left(
\frac{p'+1}{p+1}
\right)^2
(p+1)^2
\sin^2(\delta/2)
C_\beta
\\
&=
c'
\left(
\frac{p'+1}{p+1}
\right)^2
x_p \ge
c'x_0
\left(
\frac{p'+1}{p+1}
\right)^2.
\end{aligned}
\]
The shot bound then follows directly from
\[
q_0^{(p')}
\ge
\frac{x_{p'}}{1+x_{p'}}
\]
and the certified budget
\[
S_{\mathrm{cert}}^{(p')}
=
\left\lceil
\left(
1+\frac{1}{x_{p'}}
\right)
\ln\frac{1}{\epsilon}
\right\rceil.
\]
\end{proof}

The following instance-dependent policies can stabilize the filtering
schedule in practice.

\begin{enumerate}[leftmargin=1.6em]

\item
\emph{Optimising \(\boldsymbol\beta\) can increase \(C_\beta\).}
The mixer-envelope mass on the optimum set is
\[
C_\beta
=
\sum_{z\in\Omega^\star}
W_p(z;\boldsymbol\beta).
\]
Gradient-free parameter search or local refinement over
\(\boldsymbol\beta\) can be used to increase this mass before the
Fej\'er reweighting.

\item
\emph{Optimising \(\gamma\) can enlarge the phase gap \(\delta\).}
With
\[
\theta(z)
=
\gamma E(z)
\pmod{2\pi},
\]
a one-dimensional sweep in \(\gamma\) can avoid aliasing among
near-optimal energies and increase
\[
\delta
=
\min_{y\notin\Omega^\star}
\operatorname{dist}_{\mathbb T}
\bigl(\theta(y),\theta^\star\bigr).
\]
Because
\[
x
=
(p+1)^2
\sin^2(\delta/2)
C_\beta
\]
grows quadratically in \(\sin(\delta/2)\) at fixed \(p\), a larger
\(\delta\) improves the ratio bound and permits smaller filter order
without invoking the approximation
\(\sin(\delta/2)\approx\delta/2\).

\item
\emph{Normalization and schedule stabilization.}
Scaling \(H_C\mapsto\alpha H_C\) rescales the corresponding phase angle
as \(\gamma\mapsto\gamma/\alpha\). Normalize the diagonal cost by
choosing
\[
\widehat H_C
:=
\frac{H_C}{R_{\mathrm{op}}},
\qquad
R_{\mathrm{op}}
\ge
\|H_C\|,
\]
so that
\[
\|\widehat H_C\|
\le 1.
\]
The consistent wrap condition is
\[
p\,\widehat\gamma_{\mathrm{safe}}\,
\|\widehat H_C\|
\le
\pi.
\]
A sufficient choice is therefore
\[
\widehat\gamma_{\mathrm{safe}}
=
\frac{\pi}{p}.
\]

Alternatively, choose a coefficient-sum bound satisfying
\[
\|H_C\|
\le
R_{\mathrm{bound}}
\le
2\sum_j|c_j|,
\]
set
\[
\widehat H_C
:=
\frac{H_C}{R_{\mathrm{bound}}},
\qquad
\widehat\gamma
:=
\gamma R_{\mathrm{bound}},
\]
and impose
\(
p\,\widehat\gamma\,
\|\widehat H_C\|
\le
\pi.
\)
This prevents the normalization scale from producing a vanishing base
step and keeps
\(
\delta
=
\widehat\gamma\Delta_{\mathrm{lat}}
\) finite.

\end{enumerate}


\section{Outlook and Conclusion}
\label{sec:outlook}

\subsection{Coherent realizations of positive spectral filters}

The main-text guarantees are derived in an analytic reference model,
while Appendix~\ref{app:coherent-fejer} shows that the same Fej\'er
weighting can be realized coherently through a postselected spectral
filter. The outstanding question is whether one can construct fully
coherent, non-postselected, and hardware-efficient unitary
implementations whose computational-basis statistics retain comparable
off-peak suppression. A natural route is to realize nonnegative trigonometric polynomials
\[
P(\theta)
=
\sum_{m=-p}^{p}a_m e^{im\theta},
\qquad
a_{-m}=a_m\ge0,
\qquad
P(\theta)\ge0.
\]
as polynomial transformations of the phase unitary
$U_C(\gamma)=e^{-i\gamma H_C}$ using methods like Quantum Signal Processing (
QSP)~\cite{BerryEtAl2015Taylor,LowChuang2017QSP}. %
The coherent construction in Appendix~\ref{app:coherent-fejer} establishes that the positive Fej\'er weighting exposed
by the analytic reference model admits a coherent physical realization in
principle. The key challenge is to eliminate postselection while keeping the ancilla, gate, and constant-factor overheads under control. Identifying near-term-friendly, non-postselected implementations of this kind is an open problem suggested by this work.
 As a result, questions of postselection acceptance, ancilla count,
controlled-unitary depth, circuit compilation, and architecture
optimization are therefore downstream properties of the chosen
implementation. 

A complementary route is provided by subsequent work that separates mixer-induced amplitude transport from coherent interference
\cite{OnahHadfieldSeparatingGeometry}. This decomposition provides a
way to move beyond the dephased reference model by controlling the
transport and phase-alignment contributions directly within the fully
coherent dynamics. Whether this route can yield a non-postselected
coherent analogue of the present finite-depth, finite-shot Fej\'er guarantee remains an open question.

Once a coherent realization has been specified, recent circuit
optimization and quantum-architecture-search methods may provide
useful downstream tools for reducing its implementation cost.
AlphaTensor--Quantum uses reinforcement-learning-assisted tensor
decomposition to reduce the \(T\)-count of circuits expressed in the
Clifford+\(T\) gate set~\cite{RuizEtAl2025AlphaTensor}. Such methods
could be applied after decomposing a coherent spectral-filter
construction into fault-tolerant primitives. At the variational and
hardware-adaptation level, topology-driven quantum architecture search
first selects promising connectivity patterns before optimizing gate
types~\cite{SuEtAl2025TopologyQAS}, while predictor-based architecture
search using ZX-calculus exploits equivalence-preserving graphical
rewrites to identify alternative circuit structures with fewer
candidate evaluations~\cite{LiEtAl2025ZXQAS}.


\subsection{Conclusion}

We developed a finite-resource perspective on constrained quantum optimization algorithms in which circuit depth and shot count are treated as first-class budgets. On the CE--QAOA kernel, lattice-normalized cost phases expose a positive Fej\'er reweighting mechanism for the phase unitary $U_C(\gamma)=e^{-i\gamma H_C}$. In our analyses, this mechanism yields dimension-free lower bounds on success probability through the mixer-envelope mass on the target set and a wrapped phase-gap proxy controlling off-peak suppression, which are two instance-facing quantities. The same framework also closes the feasibility story quantitatively. When the filtering analysis is specialized to penalty-only phases, it gives explicit finite-depth lower bounds on feasible mass, complementing the structural reachability statement on the invariant sector developed in Sec. \ref{sec:feasibility-from-finite-level-transition} and Thm. \ref{thm:finite-depth-feasible}.

Conceptually, the Fej\'er viewpoint separates exploration from selection. The mixer provides a conservative exploration baseline
through the envelope $W_p(\cdot;\boldsymbol\beta)$, while the Fej\'er
kernel supplies a positive selection rule that amplifies the target
phase and suppresses off-target phases. More importantly,
Theorem~\ref{thm:dimension-free-success} and
Eq.~\eqref{eq:finite-synth} expose a finite-resource compensation law
in which filter order, phase separation, mixer-envelope mass, and shot
count trade against one another. Since for every $\delta>0$ and
$C_\beta>0$ the guarantee admits a finite depth--shot budget, a small
phase gap or mixer-envelope mass can be offset by increasing the filter
order and, when necessary, the number of shots. More generally, the
combined control parameter need not remain constant with problem size.
If
\[
(p+1)^2\sin^2(\delta/2)C_\beta
=
\Omega(n^{-k})
\]
for any constant $k\geq 0$, then Eq.~\eqref{eq:finite-synth} gives
\[
S_{\mathrm{cert}}
=
O\!\left(
n^k\ln\frac{1}{\epsilon}
\right).
\]
Thus even inverse-polynomial degradation of the combined
depth--gap--envelope factor preserves a polynomial finite-shot
guarantee. 

Finally, Appendix~\ref{app:coherent-fejer} shows that the same Fej\'er weighting can be realized at the coherent level through a postselected spectral filter. In light of near term limitations in quantum computation, the main open problem left by this work is therefore to replace this principle-level coherent realization by non-postselected, hardware-efficient unitary constructions with comparable filtering power and controlled overhead.

\subsection*{Data Availability} All data in this paper are contained in the paper. 


\begin{appendix}

\section{Finite Feasibility Bounds from Fej\'er Filtering}
\label{sec:fejer-feasibility}

\paragraph{Feasibility as a phase-selection problem.}
To extract a \emph{feasibility} guarantee from the Fej\'er factorization law, we
apply the same filtered-reference construction as in \S\ref{sec:fejer}, but we choose
the \emph{phase signal} to be the penalty Hamiltonian.
Concretely, set
\[
H_C \equiv H_{\mathrm{pen}},
\qquad
\theta(z):=\gamma\,H_{\mathrm{pen}}(z)\ (\mathrm{mod}\ 2\pi),
\qquad
\theta^\star:=0,
\]
so that the target phase $\theta^\star$ corresponds exactly to the feasible level set
\(
L_0=\{z:\,H_{\mathrm{pen}}(z)=0\}.
\)
Under the filtered-reference model, the measurement law reads
\begin{equation}
\label{eq:fejer-feasibility-factor}
\Pr^{\mathrm{ref}}_{p}[z]
=
\frac{W_p(z;\boldsymbol\beta)\,F_p(\theta(z))}
{\sum_{y\in\mathcal X} W_p(y;\boldsymbol\beta)\,F_p(\theta(y))},
\end{equation}
and the reference-model feasibility probability is
\[
\pi_{\mathsf F}^{\mathrm{ref}}(p;\boldsymbol\beta,\gamma)
\;:=\;
\Pr^{\mathrm{ref}}_{p}[L_0]
\;=\;
\sum_{z\in L_0}\Pr^{\mathrm{ref}}_{p}[z].
\]

Define the feasible envelope mass
\begin{equation}
\label{eq:Cbeta-feasible}
C_{\beta}^{\mathsf F}
\;:=\;
\sum_{z\in L_0} W_p(z;\boldsymbol\beta)
\;\in\;(0,1].
\end{equation}
This is the analogue of the weight in \S\ref{sec:base-success-three-regimes},
with the ``optimal set'' replaced by the feasible set. 

\paragraph{Penalty-phase separation.}
Since $H_{\mathrm{pen}}(z)\in\{0,1,\dots,t_{\max}\}$ on $\OH$, the set of possible
penalty phases is $\{\gamma t\ (\mathrm{mod}\ 2\pi): t\in\mathcal V\}$, where
$\mathcal V:=\{t:\,|L_t|>0\}$ as in \S\ref{sec:feasibility-from-finite-level-transition}.
Let
\begin{equation}
\label{eq:delta-feasible}
\delta_{\mathsf F}
\;:=\;
\min_{t\in\mathcal V\setminus\{0\}}
\dist_{\mathbb T}(\gamma t,0)
\;\in\;(0,\pi],
\qquad
\dist_{\mathbb T}(\phi,\varphi):=\min_{k\in\mathbb Z}|\phi-\varphi+2\pi k|.
\end{equation}
A convenient anti-aliasing regime is
\begin{equation}
\label{eq:gamma-safe-pen}
0<\gamma\le \frac{\pi}{t_{\max}},
\end{equation}
which avoids wrap-around since then $\gamma t\in[0,\pi]$ for all $t\in\{0,1,\dots,t_{\max}\}$.
In that case, if the smallest nonzero active level is $t_{\min}:=\min(\mathcal V\setminus\{0\})$,
then $\delta_{\mathsf F}=\gamma\,t_{\min}$ (and in particular $\delta_{\mathsf F}=\gamma$ whenever
$1\in\mathcal V$).

Fej\'er feasibility bound becomes straightforward. Let $M_p(\delta)$ be the wrapped off-peak maximum defined in
Eq.~\eqref{eq:Mp-wrapped}. Since $F_p(0)=p+1$ on $L_0$ and
\[
\operatorname{dist}_{\mathbb T}(\theta(y),0)\ge\delta_F
\qquad
\text{for all }y\notin L_0,
\]
Lemma~\ref{lem:fejer-offpeak} gives
$F_p(\theta(y))\le M_p(\delta_F)$. The same denominator decomposition as in the optimality analysis yields the following feasibility guarantee in the same model.

\begin{corollary}[Dimension-free feasibility bound from Fej\'er filtering]
\label{thm:fejer-feasibility}
Assume the factorized law~\eqref{eq:fejer-feasibility-factor} with $H_C\equiv H_{\mathrm{pen}}$
and penalty-phase separation $\delta_{\mathsf F}>0$ from~\eqref{eq:delta-feasible}.
Then
\begin{equation}
\label{eq:fejer-feasibility-bound}
\pi_{\mathsf F}^{\mathrm{ref}}(p;\boldsymbol\beta,\gamma)
\;\ge\;
\frac{(p+1)\,C_{\beta}^{\mathsf F}}
{(p+1)\,C_{\beta}^{\mathsf F}+M_p(\delta_{\mathsf F})\,(1-C_{\beta}^{\mathsf F})}.
\end{equation}
Moreover, using the uniform off-peak bound $M_p(\delta)\le\bigl((p+1)\sin^2(\delta/2)\bigr)^{-1}$
from Lemma~\ref{lem:fejer-offpeak}, one obtains the explicit estimate
\begin{equation}
\label{eq:fejer-feasibility-ratio-correct}
\pi_{\mathsf F}^{\mathrm{ref}}(p;\boldsymbol\beta,\gamma)
\;\ge\;
\frac{x_{\mathsf F}}
{x_{\mathsf F}+(1-C_{\beta}^{\mathsf F})},
\qquad
x_{\mathsf F}:=(p+1)^2\,\sin^2(\delta_{\mathsf F}/2)\,C_{\beta}^{\mathsf F}.
\end{equation}
In particular, since $1-C_{\beta}^{\mathsf F}\le 1$, this implies the simpler bound
\begin{equation}
\label{eq:fejer-feasibility-ratio-simpler}
\pi_{\mathsf F}^{\mathrm{ref}}(p;\boldsymbol\beta,\gamma)
\;\ge\;
\frac{x_{\mathsf F}}{1+x_{\mathsf F}}.
\end{equation}
\end{corollary}

\begin{corollary}[Depth-$p=1,2$ feasibility bound]
\label{cor:fejer-feasibility-p1}
Under the results of Cor.~\ref{thm:fejer-feasibility} with $p=1$,

 \[
\pi_{\mathsf F}^{\mathrm{ref}}(1;\boldsymbol\beta,\gamma)
\;\ge\;
\frac{x_{\mathsf F,1}}{x_{\mathsf F,1}+(1-C_{\beta}^{\mathsf F})}
\;\ge\;
\frac{x_{\mathsf F,1}}{1+x_{\mathsf F,1}},
\qquad
x_{\mathsf F,1}:=4\sin^2(\delta_{\mathsf F}/2)\,C_{\beta}^{\mathsf F}.
\]
With $p=2$,
\[
\pi_{\mathsf F}^{\mathrm{ref}}(2;\boldsymbol\beta,\gamma)
\;\ge\;
\frac{x_{\mathsf F,2}}{x_{\mathsf F,2}+(1-C_{\beta}^{\mathsf F})}
\;\ge\;
\frac{x_{\mathsf F,2}}{1+x_{\mathsf F,2}},
\qquad
x_{\mathsf F,2}:=9\sin^2(\delta_{\mathsf F}/2)\,C_{\beta}^{\mathsf F}.
\]
\end{corollary}

\section{Coherent Fej\'er Filtering and a Dimension-Free Success Bound}
\label{app:coherent-fejer}

Here we show that it is possible in principle to obtain a Fej\'er weighting at the level of the measured computational-basis distribution, without inserting a dephasing channel. Thus, the reference-model Fej\'er factorization can be upgraded to a fully coherent
statement by implementing the Dirichlet polynomial as a \emph{coherent} spectral filter on the cost unitary exhibiting the sense in which the positive-kernel mechanism can exist at the circuit level.

\begin{definition}[Coherent Dirichlet filter centered at the optimal phase]
\label{def:coherent-dirichlet-filter}
Fix an optimal wrapped phase $\theta^\star\in[0,2\pi)$ and define
\begin{equation}
\label{eq:coherent-dirichlet-filter}
D_p^\star(H_C)
\;:=\;
\frac{1}{\sqrt{p+1}}
\sum_{r=0}^{p} e^{-ir(\gamma H_C-\theta^\star)}.
\end{equation}
Equivalently, on any computational-basis eigenstate $\ket{z}$ of $H_C$ with
$H_C\ket{z}=E(z)\ket{z}$ and $\theta(z):=\gamma E(z)\ (\mathrm{mod}\ 2\pi)$,
\begin{equation}
\label{eq:coherent-dirichlet-eigenvalue}
D_p^\star(H_C)\ket{z}
=
\frac{1}{\sqrt{p+1}}
\sum_{r=0}^{p} e^{-ir(\theta(z)-\theta^\star)}\,\ket{z}.
\end{equation}
Hence
\begin{equation}
\label{eq:coherent-fejer-diagonal}
\bra{z}\,D_p^\star(H_C)^\dagger D_p^\star(H_C)\,\ket{z}
=
F_p\!\bigl(\theta(z)-\theta^\star\bigr),
\end{equation}
where $F_p$ is the Fej\'er kernel.
\end{definition}

\paragraph{Coherent filtered protocol.}
Let $\ket{\psi_{\mathrm{env}}}$ be any coherently prepared state on $\OH$. Define the postselected filtered state
\begin{equation}
\label{eq:coherent-filtered-state}
\ket{\psi_{\mathrm{filt}}}
\;:=\;
\frac{D_p^\star(H_C)\ket{\psi_{\mathrm{env}}}}
{\sqrt{\bra{\psi_{\mathrm{env}}}D_p^\star(H_C)^\dagger D_p^\star(H_C)\ket{\psi_{\mathrm{env}}}}}.
\end{equation}
Measuring $\ket{\psi_{\mathrm{filt}}}$ in the computational basis yields
\begin{equation}
\label{eq:coherent-filtered-law}
\Pr_{\mathrm{coh}}[z]
=
\frac{|\langle z\mid\psi_{\mathrm{env}}\rangle|^2\,F_p(\theta(z)-\theta^\star)}
{\sum_{y\in\OH}|\langle y\mid\psi_{\mathrm{env}}\rangle|^2\,F_p(\theta(y)-\theta^\star)}.
\end{equation}

For
\[
\rho_{\mathrm{env}}
:=
|\psi_{\mathrm{env}}\rangle
\langle\psi_{\mathrm{env}}|,
\qquad
W_p(z;\boldsymbol\beta)
:=
|\langle z|\psi_{\mathrm{env}}\rangle|^2,
\]
define the physical filter Kraus operator
\[
K_p^\star(H_C)
:=
\frac{D_p^\star(H_C)}{\sqrt{p+1}}
=
\frac{1}{p+1}
\sum_{r=0}^{p}
e^{-ir(\gamma H_C-\theta^\star I)}.
\]
Since $0\le F_p(\vartheta)\le p+1$, this operator satisfies
$K_p^\star(H_C)^\dagger K_p^\star(H_C)\preceq I$.

Let
\begin{equation}
\label{eq:coherent-normalization}
\mathcal Z_p
:=
\sum_{z\in\mathcal H_{\mathrm{OH}}}
W_p(z;\boldsymbol\beta)
F_p(\theta(z)-\theta^\star).
\end{equation}
The postselection acceptance probability is
\begin{equation}
\label{eq:postselection-acceptance}
P_{\mathrm{acc}}
=
\operatorname{Tr}\!\left[
K_p^\star(H_C)\rho_{\mathrm{env}}
K_p^\star(H_C)^\dagger
\right]
=
\frac{\mathcal Z_p}{p+1}.
\end{equation}
Consequently, obtaining $S$ accepted filtered samples requires
$S/P_{\mathrm{acc}}$ filter executions on average. If one execution,
including state preparation and the controlled filter operations, has
gate cost $G_{\mathrm{att}}$ and depth $d_{\mathrm{att}}$, the expected
end-to-end resources for $S$ accepted samples are
\begin{equation}
\label{eq:accepted-sample-overhead}
\mathbb E[G_{\mathrm{tot}}]
=
\frac{S\,G_{\mathrm{att}}}{P_{\mathrm{acc}}},
\qquad
\mathbb E[d_{\mathrm{tot}}]
=
\frac{S\,d_{\mathrm{att}}}{P_{\mathrm{acc}}}.
\end{equation}

To relate the conditional guarantee to the probability per raw
execution, define the target mass in the input state by
\begin{equation}
\label{eq:coherent-envelope-target-mass}
C_{\mathrm{env}}
:=
\sum_{z\in\Omega^\star}
W_p(z;\boldsymbol\beta).
\end{equation}
Because $F_p(0)=p+1$, the success probability conditioned on acceptance
is
\begin{equation}
\label{eq:conditional-coherent-success}
q_{\mathrm{coh}}
:=
\Pr_{\mathrm{coh}}
[\Omega^\star\mid\mathrm{acc}]
=
\frac{(p+1)C_{\mathrm{env}}}{\mathcal Z_p}.
\end{equation}
Combining Eqs.~\eqref{eq:postselection-acceptance} and
\eqref{eq:conditional-coherent-success} gives the end-to-end identity
\begin{equation}
\label{eq:end-to-end-coherent-success}
P_{\mathrm{acc}}\,q_{\mathrm{coh}}
=
C_{\mathrm{env}}.
\end{equation}
Thus this postselected realization increases the target probability
within the accepted ensemble, but the probability that a raw execution
is both accepted and optimal is $C_{\mathrm{env}}$. The expected number
of raw executions required to obtain one accepted optimal sample is
therefore $C_{\mathrm{env}}^{-1}$. More generally, for
$0<C_{\mathrm{env}}<1$, a sufficient number of raw executions for
obtaining at least one accepted optimum with failure probability at
most $\epsilon$ is
\begin{equation}
\label{eq:end-to-end-shot-overhead}
N_{\mathrm{raw}}
\ge
\left\lceil
\frac{\log\epsilon}
{\log(1-C_{\mathrm{env}})}
\right\rceil.
\end{equation}
When the coherent input has the same diagonal envelope as the reference
model, $C_{\mathrm{env}}=C_\beta$.

\subsection*{Conflict of Interests}
All authors declare no competing interests.

\subsection*{Author contributions}
C.O. wrote the main paper and prepared all the figures, K.M. provided guidance and interpretation of results during the project. All the authors reviewed the paper.

\subsection*{Funding}
This research received no external funding.

\subsection*{Additional information}
Correspondence and requests for materials should be addressed to C.O.\\ (\texttt{conah@atom-computing.com}).

\end{appendix}

\printbibliography
\end{document}